\documentclass[10pt,journal]{IEEEtranTCOM}

\IEEEoverridecommandlockouts	
\usepackage[dvips]{graphicx}
\usepackage{amssymb}
\usepackage{cite}
\usepackage{subfigure}
\usepackage{multirow}
\usepackage{pifont}
\usepackage{hyperref}
\usepackage{bm}
\usepackage{amsmath}
\usepackage{amsthm}
\usepackage{stfloats}
\usepackage{float}

\usepackage{subeqnarray}
\usepackage{algorithm}
\usepackage{algorithmic}
\usepackage{color}
\usepackage{cases}
\usepackage{caption2}
\usepackage{url}
\usepackage{array}
\usepackage{enumerate}
\usepackage{tabularx}
\usepackage{multirow}
\usepackage[table,xcdraw]{xcolor}
\usepackage{diagbox}
\usepackage{booktabs}
\usepackage{colortbl}
\makeatletter

\newcommand{\Rmnum}[1]{\expandafter\@slowromancap\romannumeral #1@}
\makeatother
\renewcommand{\baselinestretch}{1}

\theoremstyle{plain}
\newtheorem{theorem}{\textit{Theorem} \textit}

\newtheorem{corollary}{Corollary}

\theoremstyle{definition}

\theoremstyle{remark}
\newtheorem{remark}{\textbf{Remark}\bf }

\begin{document}
\title{MSE Minimization for Movable Antenna Enhanced Wireless Sensing via Array Geometry Optimization}
\author{Wenrong Chen, \IEEEmembership{Member, IEEE}, Lipeng Zhu, \IEEEmembership{Senior Member, IEEE}, Wenyan Ma, \IEEEmembership{Member, IEEE}, Zhi Chen, \IEEEmembership{Senior Member, IEEE}, and Rui Zhang, \IEEEmembership{Fellow, IEEE}
	\thanks{This work was supported in part by the National Key R\&D Program of China under Grant 2024YFE0200400 and the National Natural Science Foundation of China (NSFC) under grant number 62671044.
	}
	\thanks{
	\IEEEcompsocthanksitem W. Chen and Z. Chen are with the National Key Laboratory of Wireless Communications, University of Electronic Science and Technology of China, Chengdu 611731, China (E-mail: wenrongchen@std.uestc.edu.cn; chenzhi@uestc.edu.cn).
	\IEEEcompsocthanksitem L. Zhu is with the State Key Laboratory of CNS/ATM and the School
	of Interdisciplinary Science, Beijing Institute of Technology, Beijing 100081, China (E-mail: zhulp@bit.edu.cn). 
	\IEEEcompsocthanksitem W. Ma and R. Zhang are with the Department of Electrical and Computer Engineering, National University of Singapore, Singapore 117583 (E-mail: wenyan@u.nus.edu; elezhang@nus.edu.sg).
	}
}
\maketitle

\begin{abstract}
This paper develops adaptive movable-antenna (MA) array structures to enhance wireless sensing performance under arbitrary signal-to-noise ratio (SNR) conditions, using the angle-of-arrival (AoA) estimation mean square error (MSE) as the sensing metric. 
Due to the presence of strong sidelobes in the ambiguity function of sparse MA arrays, the classical Cram\'er-Rao bound (CRB) of AoA estimation MSE fails to characterize the true MSE in the low and medium SNR regimes, rendering CRB-optimized antenna position vector (APV) designs suboptimal for MA-enabled sensing.
To overcome this limitation, we derive a deterministic MSE approximation that captures the local estimation error within the mainlobe of the ambiguity function and the outlier error induced by sidelobes, providing a tractable and accurate optimization objective.
We then propose three adaptive MA structures that are tailored to different hardware and computational requirements. 
For systems with fully independent motors, a movable free array (MFA) is optimized by a two-stage greedy algorithm, which approaches the performance of the globally optimal exhaustive search but with lower computational complexity.
Halving the required motors, a movable nested array (MNA) fixes the dense inner subarray and adjusts only the outer-subarray spacing, reducing the optimization to a one-dimensional search. 
As the most hardware-efficient design, a movable twin array (MTA) is proposed to adjust the inter-subarray spacing only with a single motor, while still achieving significant MSE reduction over fixed array geometries. 
Numerical results show that the derived approximation closely matches the true MSE compared with existing performance bounds, and that the proposed structures outperform the CRB-optimized designs by orders of magnitude in MSE across the tested low-to-medium SNR range.
Moreover, it reveals that in the low-SNR regime, APV reconfiguration is more effective for improving estimation accuracy than simply increasing the number of fixed antennas. 
\end{abstract}


\IEEEpeerreviewmaketitle

\section{Introduction}

Next-generation wireless networks increasingly integrate high-accuracy sensing to support location-aware applications such as autonomous driving and industrial automation\cite{you2025next,liu2022survey}.
Among assorted sensing tasks, angle-of-arrival (AoA) estimation serves as a key function, whose accuracy is quantified by the mean square error (MSE) and is governed by the antenna array geometry\cite{xiao2024channel}. 
Enhancing AoA estimation accuracy calls for a large-aperture array, but the associated hardware complexity and power consumption scale directly with the number of antennas\cite{lu2014overview,wang2024tutorial}.
Sparse arrays address this issue by enlarging the inter-antenna spacing to attain a large aperture, thereby achieving the angular resolution of a much larger dense array with only a small number of elements\cite{roberts2011sparse,cen2009linear}.
However, such wide spacing inevitably introduces prominent sidelobes (also called grating lobes) of the ambiguity function, whose resulting ambiguity becomes increasingly pronounced as the signal-to-noise ratio (SNR) decreases, severely degrading the achievable MSE in the low-to-medium SNR regime\cite{wang2023can}.
Such sparse arrays are conventionally realized with fixed-position antennas (FPAs) whose static geometry cannot adapt to the SNR conditions, thus limiting their achievable sensing performance\cite{zhu2026tutorial}.

Movable antenna (MA) overcomes this limitation by enabling flexible antenna position adjustment according to the varying SNR conditions, thereby outperforming traditional FPAs\cite{zhu2024movabletcm,ma2023mimo}.
Exploiting this spatial mobility, an MA system reconfigures a sparse set of antennas into a large aperture within a designated movement region, thereby enhancing angular resolution beyond that of FPAs for improving communication performance \cite{zhu2024movable,mei2024movable,zhu2024movabletwc}. 
In addition, antenna array rotation can be jointly exploited with antenna positioning to further enhance the system performance\cite{shao20256dma,shao20246d,shao20256d}.
Crucially, this geometry can either be pre-configured offline for a specific deployment or be adapted online to the operating SNR, thus providing MA with a spatial DoF for sensing that FPAs cannot offer\cite{ma2024movable,wang2025movable}.

The key to fully leveraging the potential of MAs lies in optimizing the antenna position vector (APV), which shapes the spatial sampling pattern of the incident wavefront and thereby determines the achievable AoA estimation accuracy\cite{mao2025movable,ma2026movabletsp}.
Traditionally, the Cram\'er-Rao bound (CRB) is widely adopted as a tractable optimization objective for APV design owing to its closed form and asymptotic optimality at high SNR\cite{song2023intelligent,ma2024movable}.
However, the CRB characterizes only the local estimation error within the mainlobe and cannot capture the sidelobe-induced outlier errors inherent to sparse arrays\cite{kay1993fundamentals,van2002modulation}.
At low-to-medium SNR, such sidelobe ambiguities can cause the estimator to lock onto an incorrect outlier peak, incurring large errors that dominate the actual MSE and render the CRB unattainable\cite{khan2010analysis,athley2001optimization}.
CRB-based APV designs therefore become increasingly suboptimal as the SNR degrades\cite{khan2009explicit}.
This calls for both a more accurate MSE expression as an optimization objective, and a novel APV design suited to the low-SNR regime sensing.

Comparative studies have shown that the preferred array geometry varies with the SNR condition\cite{athley2007radar}.
At high SNR, the MSE-minimizing APV is a two-subarray structure in which the antennas aggregate at both ends of the movable region to form two split uniform linear arrays (SULA), maximizing the effective aperture\cite{athley2000direction,ma2024movable}.
At low-to-medium SNR, sidelobe-suppressing geometries are favoured instead\cite{athley2003space}.
The minimum redundancy array (MRA) maximizes the spatial degrees of freedom by minimizing redundant covariance lags, but admits no closed-form construction and can only be obtained via extensive numerical search\cite{moffet1968minimum}.
As a closed-form alternative toward the same end, the nested array (NA) offers a near-optimal structure by concatenating a co-designed dense inner and sparse outer subarray\cite{pal2010nested}.
Since no single static geometry remains optimal across a wide SNR range, the array geometry needs to adapt to the operating SNR, motivating the use of MA for this purpose.
Along this line, the state-of-the-art MA design in \cite{ma2026movable} minimizes the CRB subject to a maximum sidelobe-level constraint via successive convex approximation (SCA) at each SNR.
However, it still lacks an explicit MSE-based objective and requires fully independent motors for every antenna to realize this irregular repositioning, incurring high hardware complexity.

To characterize the MSE more accurately than the CRB, tighter bounds have been developed, including the Weiss-Weinstein bound (WWB), the Ziv-Zakai bound (ZZB), and the explicit ZZB (EZZB)\cite{ziv1969some,weiss1985lower,bell1997extended,bell1996explicit,vertatschitsch1991impact,khan2013explicit,athley2005threshold}. 
Derived from the covariance inequality, the WWB provides a tight lower bound in the low-SNR regime, but its evaluation relies on a set of randomly selected test points that admits no deterministic expression\cite{weiss1985lower}.
The ZZB instead relates the MSE to the error probability of a binary hypothesis test, thereby inherently accounting for sidelobe ambiguities\cite{bell1997extended}, albeit at the cost of tedious numerical integration\cite{bell1996explicit}.
This tractability issue is partially resolved by the EZZB, which  decomposes the MSE into the a priori random error, the mainlobe local error, and the sidelobe-induced error\cite{khan2013explicit}.
Nevertheless, being lower bounds rather than the achievable MSE, their tightness depends on the APV-dependent sidelobe characteristics, since the sidelobe contribution is lower-bounded through the envelope of the ambiguity function rather than its actual structure\cite{athley2005threshold}.
This motivates the derivation of a deterministic MSE approximation that is both accurate and tractable as the objective for APV optimization. 

In this paper, we propose adaptive MA structures for AoA estimation whose APV is dynamically reconfigured across the low-to-high SNR regime, spanning three approaches tailored to different levels of motor independence and computational complexity. 
The main contributions are as follows:


1) We derive a deterministic MSE approximation that jointly captures the local estimation error within the mainlobe of the ambiguity function and the outlier error induced by sidelobes, remaining accurate across the low-to-high SNR regime where the CRB is no longer valid.
This approximation serves as the tractable optimization objective underlying the three adaptive APV designs developed in this paper. 

2) Building on this approximation, we develop three adaptive MA structures that achieve different trade-offs between computational and hardware complexities to enhance AoA estimation accuracy at low-to-medium SNR.

\begin{itemize}

\item \textit{Movable Free Array (MFA)}. 
For systems retaining full spatial degree of freedom, we propose a two-stage greedy algorithm that optimizes the APV into the MFA, approaching the globally optimal exhaustive search while reducing the computational complexity from an exponential to a polynomial order.

\item \textit{Movable Nested Array (MNA)}. MNA fixes the dense inner subarray and adjusts only the outer-subarray spacing, thus halving the required motors and reducing the design to a one-dimensional search while attaining MSE comparable to the MFA.

\item \textit{Movable Twin Array (MTA)}. MTA adjusts only the inter-subarray spacing with a single motor, attaining significant MSE reduction over FPAs at the lowest hardware cost. 
For this structure, we derive a closed-form MSE expression that reveals its structural trade-off between sidelobe suppression and aperture maximization.
\end{itemize}





3) Numerical results show that the MFA and MNA closely approach the performance of the globally optimal exhaustive search, with MNA requiring only half the motors, while the single-motor MTA still attains orders-of-magnitude MSE reduction over FPAs.
Moreover, SNR-adaptive APV reconfiguration proves more effective for MSE reduction than simply increasing the number of fixed antennas in the low-to-medium SNR regime.

The rest of the paper is organised as follows.
Section II introduces the MA system model for AoA estimation. 
Section III presents the derivation of the MSE approximation, incorporating both local errors and sidelobe-induced outliers. 
Section IV describes the three proposed adaptive MA geometries, including the MFA, the MNA, and the MTA.
Section V provides extensive numerical results and performance analysis to validate the proposed MSE model and adaptive structures. 
Finally, the paper is concluded in Section VI.

\section{System Model and Problem Formulation}
\begin{figure}[tp]
	\centering
	\includegraphics[width=0.9\linewidth]{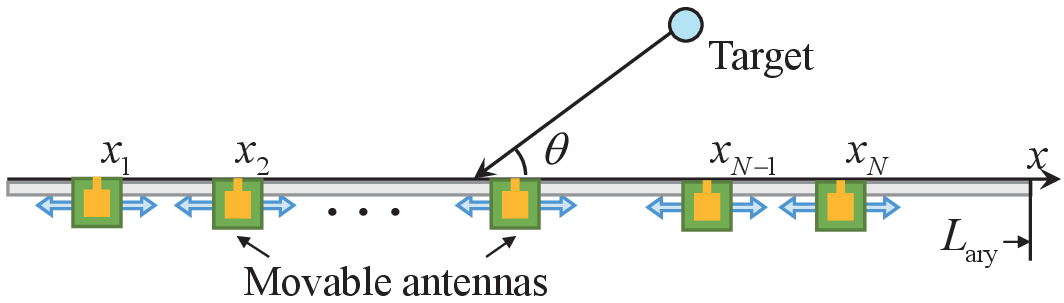}
	\caption{One-dimensional MA array for target angle estimation.}
	\label{fig_system}	
\end{figure}

We consider a linear MA array with $N$ elements to estimate a target's AoA, as illustrated in Fig.~\ref{fig_system}.
The positions of the elements can be adjusted individually with the aid of motors, and the position of the $n$-th MA is denoted by $x_n\in[0,L_{\rm{ary}}]$, where $L_{\rm{ary}}$ is the length of the region for antenna movement.
The APV of all $N$ MAs is given by
\begin{equation}
	\textbf{x}\triangleq[x_1,x_2,\ldots,x_{N}]^{\rm{T}}\in \mathbb{R}^{N}, 
\end{equation}
where $0\leq x_1<x_2<\ldots<x_N\leq L_{\rm{ary}}$.

Under far-field line-of-sight (LoS) propagation between the target and the MA array, the spatial AoA is defined as $u_0\triangleq \cos{\theta}$, where $\theta\in[0,\pi]$ is the physical steering angle of the target-receiver LoS path relative to the MA plane. 
The corresponding steering vector, as a function of the APV $\textbf{x}$ and the spatial AoA $u_0$, is written as
\begin{equation}
	{\bf{\alpha }}({\bf{x}},u_0) = {\left[ {{e^{j\frac{{2\pi }}{\lambda }{x_1}u_0}},{e^{j\frac{{2\pi }}{\lambda }{x_2}u_0}},\ldots,{e^{j\frac{{2\pi }}{\lambda }{x_N}u_0}}} \right]^{\rm T}}\in {\mathbb{C}}^{N},
	\label{array vector}
\end{equation}
where $\lambda$ is the wavelength, such that the LoS channel vector between the target and the MA array is given by
\begin{equation}
	{\bf{h}}({\bf{x}},u_0 ) = \beta {\bf{\alpha }}({\bf{x}},u_0),
	\label{channel}
\end{equation}
where $\beta$ denotes the complex fading coefficient from the target to the array origin.

Considering the use of the multiple signal classification (MUSIC) algorithm to estimate the spatial AoA $u_0$, the receiver collects $T$ signal snapshots of the reflected echoes, and the received signal at the $t$-th ($t=1,2,\ldots,T$) snapshot is given by
\begin{equation}
	{{\bf{y}}_t} = {\bf{h}}({\bf{x}},u_0){s_t} + {{\bf{z}}_t},
\end{equation}
where ${s_t}$ is the signal reflected from the target at the $t$-th snapshot with average power ${\bf{E}}[{\left| {{s_t}} \right|^2}] = {P_t}$, and ${{\bf{z}}_t}\sim{\cal C}{\cal N}(0,{\sigma ^2}{{\bf{I}}_N})$ is the additive white Gaussian noise vector.
The receive SNR is accordingly defined as ${\rm{SNR}}\triangleq{|\beta|^2}/{\sigma^2}$, representing the average per-antenna, per-snapshot SNR. 

Stacking the $T$ snapshots, the received signal matrix is written as
\begin{equation}
	{\bf{Y}} = {\bf{h}}({\bf{x}},u_0){{\bf{s}}^{\rm T}} + {\bf{Z}}, 
	\label{Y_matrix}
\end{equation}
where ${\bf{s}} = {[{s_1},{s_2},\ldots,{s_T}]^{\rm T}} \in {\mathbb{C}^{\rm T}} $ and ${\bf{Z}} = [{{\bf{z}}_1},{{\bf{z}}_2},\ldots,{{\bf{z}}_T}] \in {\mathbb{C} ^{N \times T}}$, whose columns are independently and identically distributed following ${\cal CN}(0,\sigma^2{\bf{I}}_N)$.

The sample covariance matrix is computed as 
\begin{align}
	{{\bf{R}}_{\bf{Y}}}
	&=\frac{1}{T}{\bf{Y}}{{\bf{Y}}^H}\notag
	\label{covariance}
\end{align}
from which the spatial spectrum is constructed as 
\begin{equation}
	V(\bar u) = {{\bf{\alpha }}^H}({\bf{x}},\bar u){{\bf{R}}_{\bf{Y}}}{\bf{\alpha }}({\bf{x}},\bar u),
	\label{func_V}
\end{equation}
and the AoA estimate is obtained by locating its peak, i.e.,
\begin{equation}
	\hat u = \arg \mathop {\max }\limits_{\bar u \in [ - 1,1]} V(\bar u).
	\label{def_estAOA}
\end{equation}

In this paper, we aim to design the APV that minimizes the AoA estimation MSE, given in \eqref{def_mse}, based on the noisy measurements in \eqref{Y_matrix}.
The estimation accuracy is quantified by the MSE, defined as
\begin{equation}
	{\rm{MSE}}({\bf{x}}) \triangleq {\bf{E}}\left[ {{{(\hat u - u_0)}^2}} \right].
	\label{def_mse}
\end{equation} 
Without loss of generality, we assume the first MA is at the origin, i.e., $x_1=0$, and to mitigate mutual coupling among MAs, the minimum spacing between adjacent elements is set to half wavelength.
Accordingly, the APV optimization problem aiming to minimize the MSE is formulated as
\begin{subequations}
	\label{problem}
	\begin{align}
		{\rm{(Q)}}\quad\min_{\bf{x}}\quad&
		{\rm MSE}({\bf{x}})\\
		{\rm s.t.}\quad&0=x_1<x_2<\cdots<x_N\leq L_{\rm{ary}},\label{condi_x}\\
		&x_{n+1}-x_n\geq\frac{\lambda}{2},~n=1,\ldots,N-1.
		\label{cons_space}
	\end{align}
\end{subequations}

The main challenge lies in obtaining an accurate tractable expression of \eqref{problem}.
By the property of the second-order moment\cite{cinlar2013introduction}, the MSE can be decomposed into a variance term and a bias term, i.e.,
\begin{equation}
	{\rm{MSE(}}{\bf{x}}{\rm{)}} = {\mathop{\rm var}} [\hat u] + {\left( {{\bf{E}}[\hat u] - u_0} \right)^2}, 
	\label{MSE_var}
\end{equation}
where the variance term characterizes the local estimation error within the mainlobe, while the bias term captures the outlier error incurred when the estimator locks onto a sidelobe peak. 
A common approach is to approximate \eqref{MSE_var} by the CRB, a tight lower bound on the variance term $\mathop{\rm var}[\hat u]$ that is accurate in the high-SNR regime\cite{liu2022survey}. 
However, owing to the sparse nature of MA arrays, the sidelobe-induced bias term becomes non-negligible as the SNR degrades, rendering CRB-based APV optimization ineffective in the low-to-medium SNR regime\cite{athley2000direction}. 
This motivates the tractable MSE approximation derived in the following section. 

%

\section{MSE Approximation}

To solve \eqref{problem}, we decompose \eqref{def_mse} into the mainlobe error, characterizing the local estimation variance, and the sidelobe-induced outlier error, accounting for the large deviation caused by selecting an incorrect sidelobe peak\cite{athley2005threshold}, expressed as
\begin{align}
	{\rm{MSE}}({\bf x})
	&\approx \Pr[\mathrm{mainlobe}]{\bf{E}}\left[ {{{(\hat u - u_0)}^2}\left| {\rm{mainlobe}} \right.} \right]\notag\\
	&+ \Pr [{\rm{sidelobes}}]{\bf{E}}\left[ {{{(\hat u - u_0)}^2}\left| {\rm{sidelobes}} \right.} \right].
	\label{MSE_total}
\end{align}
Since the probability of selecting a sidelobe peak is much higher than that of points in its vicinity, we approximate the sidelobe region by $N_p$ discrete peaks and treat the peak-selection events as mutually exclusive\cite{athley2005threshold}.
Denote the probability of selecting the $n$-th sidelobe peak located at $u_n$ as
\begin{equation}
	P_n\triangleq\Pr[V(u_n)>V(u_0)],
	\label{Pn}
\end{equation}
and the mainlobe probability is thus
\begin{equation}
	{P_0}\triangleq\Pr[{\rm{mainlobe}}] \approx 1-\sum_{n=1}^{N_p}P_n.
	\label{P0}
\end{equation}
Substituting \eqref{Pn}, \eqref{P0} into \eqref{MSE_total}, the MSE is approximated as
\begin{equation}
	{\rm{MSE}}({\bf{x}})\! \approx\! \left(\! {1 - \sum\limits_{n = 1}^{{N_p}} {{P_n}} } \!\right){\bf{E}}[{\epsilon_u}^2] + \sum\limits_{n = 1}^{{N_p}} {{P_n}{{({u_n} - {u_0})}^2}},
	\label{MSE_split}
\end{equation}
where ${\bf{E}}[\epsilon_u^2]$ denotes the mainlobe estimation error.
According to the tail probability representation of the second moment, ${\bf{E}}[\epsilon_u^2]$ can be derived as \cite{cinlar2013introduction}
\begin{equation}
	{\bf{E}}[{\epsilon_u}^2]\triangleq\frac{1}{2}\int_0^{{u_s}} {\Pr \left( {\left| {\hat u - {u_0}} \right| \ge \frac{\delta }{2}} \right)\delta\, d\delta },
	\label{error_main}
\end{equation}
where $u_s$ is the mainlobe width.
The closed-form evaluation of \eqref{MSE_split} is given by the following theorem.


\begin{theorem}
\label{Theorem_MSE}
For the considered linear MA array with APV ${\bf{x}}$, the MSE of AoA estimation based on $T$ snapshots under SNR can be approximated as
\begin{align}
	\!{\rm{MSE}}({\bf{x}})\approx&\left(\!1-\!\sum_{n=1}^{N_p}P_n\!\right)\!\!\left[\!\frac{T\gamma}{J}\Phi\left(\sqrt{\frac{T\gamma}{2}}\right)\!+\! J^{-1}\Gamma_{{3}/{2}}\left(\frac{T\gamma}{4}\!\right)\!\right]\!\notag\\
	&+\sum\nolimits_{n=1}^{N_p}P_n(u_n-u_0)^2 ,	
	\label{MSE_ap}
\end{align}
where $\gamma\triangleq\frac{{{N^2}{\rm{SN}}{{\rm{R}}^2}}}{{1 + N{\rm{SNR}}}}$, $\Phi (z) \triangleq\int_z^\infty  {\frac{1}{{\sqrt {2\pi } }}{e^{ - {v^2}/2}}dv} $, $\Gamma_{3/2}(\cdot)$ is the regularized lower incomplete gamma function of order $3/2$, $J \triangleq\frac{{8{\pi ^2}T\gamma }}{{N{\lambda ^2}}}{{\bf{x}}^{\rm T}}\left( {{{\bf{I}}_N} - \frac{1}{N}{\bf{1}}{{\bf{1}}^{\rm T}}} \right){\bf{x}}$ is the Fisher information, and $P_n$ denotes the probability of selecting the $n$-th sidelobe.\footnote{The closed-form expression of $P_n$ and its derivation are provided in \eqref{Pn_expression} of Appendix~\ref{appendix_MSE}.}

As ${\rm{SNR}}\to\infty$, the MSE converges to the CRB, i.e., 
\begin{equation}
	\begin{split}
		\mathop {\lim }\limits_{{\rm{SNR}} \to \infty } {\rm{MSE}}({\bf{x}}) 
		= {J^{ - 1}}= \frac{{{\lambda ^2}}}{{8{\pi ^2}\gamma }}\frac{1}{{{{\bf{x}}^{\rm T}}\left( {{{\bf{I}}_N} - \frac{1}{N}{\bf{1}}{{\bf{1}}^{\rm T}}} \right){\bf{x}}}}.
	\end{split}	
	\label{MSE_theorem_highSNR_ap}
\end{equation}
\end{theorem}

\begin{IEEEproof}
	The proof is given in Appendix \ref{appendix_MSE}.
\end{IEEEproof}

%
%
%

\begin{remark}
\label{remark}
Since $J^{-1}\triangleq {\rm{CRB}}$, Eq. \eqref{MSE_theorem_highSNR_ap} shows that at high-SNR regimes, the MSE reduces to the CRB, expressed as a function of the APV as
\begin{equation}
	{\rm{CRB}} = \frac{{{\lambda ^2}(1 + N{\rm{SNR}})}}{{8{\pi ^2}T{N^2}{\rm{SN}}{{\rm{R}}^2}}}\frac{1}{{{\rm{var}}[{\bf{x}}]}},
	\label{CRB_x}
\end{equation}
which is minimized by maximizing the spatial variance of the APV, i.e., placing antennas at the two extremities of the movement region\cite{ma2024movable}.
In contrast, at low-SNR regimes, the sidelobe-related term $P_n$ becomes non-negligible, and suppressing it prefers a densely-spaced array\cite{athley2007radar}.

This reveals a fundamental trade-off between enlarging the effective aperture and suppressing the sidelobe level across the SNR range, motivating an adaptive APV design that dynamically balances the two according to the operating SNR. 
\end{remark}

\begin{algorithm}[t]
	\caption{Two-Stage Greedy Search for MFA Optimization}
	\label{alg_greedy}
	\begin{algorithmic}[1]
		\REQUIRE MA number $N$, aperture length $L_{\rm ary}$, position resolution $\Delta d$.
		\ENSURE Optimized APV $\textbf{x}_{\rm MFA}^*$. 
		\vspace{0.3em}
		\STATE \textbf{Stage 1: Iterative Element Pruning}
		\STATE Initialize ${\bf{x}}_{{\rm{itr}}}^{(0)} \gets {\left[ {0,\Delta d,\ldots,\left( {\left\lfloor {\frac{{{L_{{\rm{ary}}}}}}{{\Delta d}}} \right\rfloor  - 1} \right)\Delta d} \right]^{\rm T}}$;
		\FOR{$i = 1$ \TO $N_{\rm all} - N$}
		\STATE Obtain ${\bf{x}}_{\rm itr}^{(i)}$ by removing one element from ${\bf{x}}_{\rm itr}^{(i-1)}$ according to \eqref{iter}; 
		\ENDFOR
		\STATE ${\bf{x}}_{\rm int} \gets {\bf{x}}_{\rm itr}^{(N_{\rm all}-N)}$;
		\vspace{0.5em}
		\STATE \textbf{Stage 2: Position Refinement}
		\STATE ${\bf{x}}_{\rm pos}^{(0)} \gets {\bf{x}}_{\rm int}$, $i\gets1$; 
		\REPEAT 
		\FOR{$n = 1$ \TO $N$}
		\STATE Update the position of the $n$-th element to \eqref{xn_i} from candidate set \eqref{candi};
		\ENDFOR
		\STATE Form ${\bf{x}}_{\rm pos}^{(i)}$ via \eqref{xpos_i}, $i\gets i+1$;
		\UNTIL{${\bf{x}}_{{\rm{pos}}}^{(i)} = {\bf{x}}_{{\rm{pos}}}^{(i - 1)}$} 
		\STATE $\textbf{x}_{\rm MFA}^* \gets {\bf{x}}_{\rm pos}^{(i)}$. 
	\end{algorithmic}
\end{algorithm}

\section{Array Geometry Optimization}

\subsection{Movable Free Array}
The MSE approximation in Theorem~\ref{Theorem_MSE} depends on the sidelobe peak locations related to the APV.
Consequently, solving  \eqref{problem} requires an exhaustive search over all feasible APVs, incurring a high computational complexity. 
To address this, we propose an MFA optimized by a two-stage search algorithm, detailed below and summarized in Algorithm~\ref{alg_greedy}. 

\subsubsection{Stage 1: Iterative Element Pruning}
The array aperture $[0,L_{\rm{ary}}]$ is discretized with a fine position resolution $\Delta d<\lambda/2$, yielding the candidate location set ${\cal A}_{\rm{all}} = [0:\Delta d:{L_{{\rm{ary}}}}]$.
The total number of  candidate positions is
\begin{equation}
	N_{\rm{all}}={\left\lfloor 
	\frac{{L_{{\rm{ary}}}}}{\Delta d} \right\rfloor }\label{nall}.
\end{equation}

We begin by densely filling the available moving region with all $N_{\rm{all}}$ candidates, i.e., ${\bf{x}}_{{\rm{itr}}}^{(0)} = {[0,\Delta d,\ldots,({N_{{\rm{all}}}} - 1)\Delta d]^{\rm T}}$, and then iteratively remove one redundant element at a time. 
At the $i$-th iteration, the $k$-th element of ${\bf{x}}_{{\rm{itr}}}^{(i)}$ that induces the smallest MSE increment is removed, i.e., 
\begin{equation}
	{\bf{x}}_{{\rm{itr}}}^{(i + 1)} = \mathop{ \arg\,\, \min }\limits_{k \in \{ 1,\ldots,{N_{{\rm{all}}}} - i\} } {\rm{MSE}}\left( {{{[{\bf{x}}_{{\rm{itr}}}^{(i)}]}_{i \in \{ 1,\ldots,{N_{{\rm{all}}}} - i\} \backslash \{ k\} }}} \right).
	\label{iter}
\end{equation}
\eqref{iter} is repeated until only $N$ elements remain, yielding the intermediate APV ${\bf{x}}_{\rm int} = {\bf{x}}_{\rm itr}^{(N_{\rm all}-N)}$, which serves as the initial input to Stage 2. 
Note that, since the pruning in \eqref{iter} does not explicitly enforce the spacing constraint \eqref{cons_space}, ${\bf{x}}_{\rm int}$ may violate it and thus requires further adjustment in Stage 2.

\subsubsection{Stage 2: Position Refinement}
In this stage, each element position is iteratively refined to further reduce the MSE while enforcing \eqref{cons_space}.
It starts from the intermediate APV output of Stage 1, expressed as
\begin{equation}
	{\bf{x}}_{{\rm{pos}}}^{(0)} = {{\bf{x}}_{{\rm{int}}}} = {[x_1^{(0)}, \ldots ,x_N^{(0)}]^{\rm T}}.
\end{equation}
At the $(i+1)$-th iteration, each element is updated in turn from $n=1$ to $N$, with all other elements held fixed, $x_n^{(i)}$ is relocated to $x_n^{(i+1)}$, the position within the feasible candidate set that minimizes the MSE, i.e.,
\begin{equation}
	x_n^{(i + 1)} = \arg \mathop {\min }\limits_{{x_n} \in {{\cal C}_{n}^i}} {\rm{MSE}}\left( {[x_1^{(i)},...x_{n - 1}^{(i)},{x_n},x_{n + 1}^{(i)}...,x_N^{(i)}]^{\rm T}} \right),
	\label{xn_i}
\end{equation}
where the candidate set ${{\cal C}_{n}^i}$ enforces the minimum spacing required by \eqref{cons_space}, expressed as
\begin{equation}
	{{\cal C}_{n}^i} = \left\{ {{x_n} \in {{\cal A}_{{\rm{all}}}}\left| {\left| {{x_n} - x_m^{(i)}} \right| \ge \frac{\lambda}{2},\forall m = \{ 1,\ldots,N\} \backslash \{ n\} } \right.} \right\}. 
	\label{candi}
\end{equation}
Sweeping through all $N$ elements via \eqref{xn_i} constitutes one iteration, producing the updated APV
\begin{equation}
	{\bf{x}}_{{\rm{pos}}}^{(i + 1)} = {[x_1^{(i + 1)}, \ldots ,x_N^{(i + 1)}]^{\rm T}}.
	\label{xpos_i}
\end{equation}
The iteration terminates once ${\bf{x}}_{{\rm{pos}}}^{(I+1)}={\bf{x}}_{{\rm{pos}}}^{(I)}$, where the APV converges to a locally optimal solution $\textbf{x}_{\rm{MFA}}^{*}={\bf{x}}_{{\rm{pos}}}^{(I+1)}$, constituting the proposed MFA.

Note that the number of iterations increases with the total number of candidate positions. 
Consequently, a large aperture or a fine position resolution increases the complexity, limiting its practicality for applications with limited computational capability.
Moreover, $\textbf{x}_{\rm{MFA}}^{*}$ lacks an explicit geometry, requiring a fully independent motor for every element and thereby increasing the hardware complexity.

To address these limitations, we further propose two low-complexity MA structures that maintain near-optimal performance while substantially reducing both complexities.

\begin{figure}[tp]
	\centering
	\includegraphics[width=\linewidth]{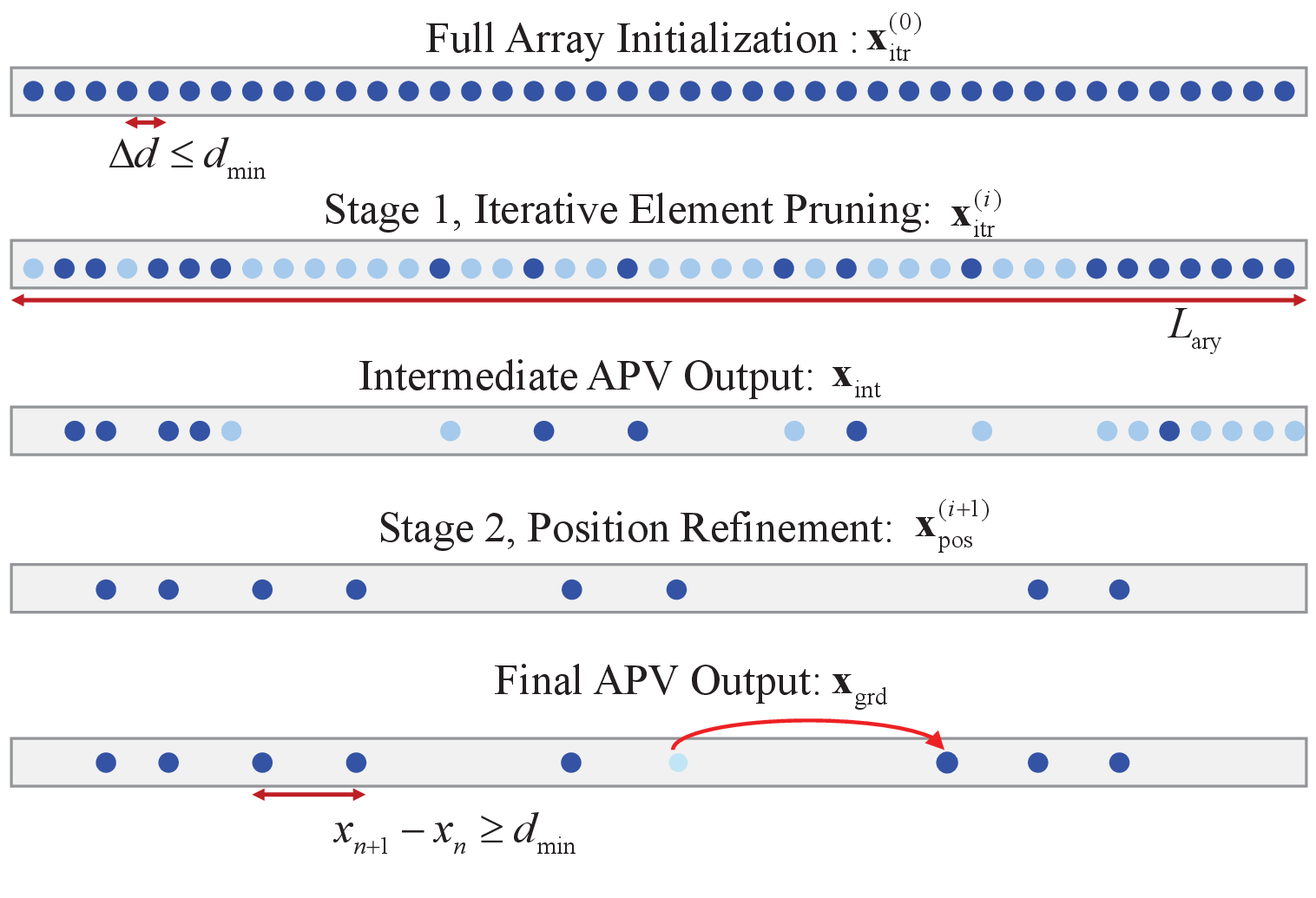}
	\caption{Illustration of the proposed Greedy-search algorithm for MFA optimization.}
	\label{fig_Greedy_search_config}	
\end{figure}

\subsection{Movable Nested Array}
\subsubsection{Design Concept}
Reducing the redundant spatial covariance lags of the inter-element spacings improves AoA estimation resolution\cite{athley2007radar}. 
The array redundancy is defined as the ratio between the total number of spacings and the number of distinct spacings ${\cal D}=\{x_i-x_j|1\le i,j\le N\}$,
\begin{equation}
	R = \frac{N(N-1)}{2L_{\rm{uni}}},\quad L_{\rm{uni}}=|{\cal D}|,
	\label{redundancy}
\end{equation}
where ${L_{{\rm{uni}}}} = \left| {\cal D} \right|$ is the number of distinct lags\cite{moffet1968minimum}.
The design goal is thus to minimize \eqref{redundancy}, i.e., to maximize the number of distinct lags $L_{\rm{uni}}$ for a given $N$\cite{moffet1968minimum}.
However, $L_{\rm{uni}}$ admits no closed-form expression in terms of the antenna positions, so the theoretically optimal geometry, the MRA\cite{moffet1968minimum}, can be obtained via exhaustive search and lacks a specific pattern, hindering its extension to adaptive MA systems. 

Towards the same objective while admitting a closed-form structure, \cite{pal2010nested} proposed a suboptimal the two-level NA, which concatenates a dense inner subarray and a sparse outer subarray.
As plotted in Fig.~\ref{fig_GNA_SPAconfig}(a), the APV of traditional NA is given by\cite{pal2010nested}
\begin{equation}
	{\bf x}_{\rm NA}=[{\bf{x}}_{\mathrm{NA},1}^{\rm T},{\bf{x}}_{\mathrm{NA},2}^{\rm T}]^{\rm T},
\end{equation}
where
\begin{subequations}
	\label{NA_geo}
	\begin{align}
		{\bf{x}}_{\mathrm{NA},1} &= \frac{\lambda }{2}\cdot {[0,1,2,\ldots,{N_1} - 1]^{\rm T}},\\
		{\bf{x}}_{\mathrm{NA},2} &= \frac{\lambda }{2}{N_1} + \frac{\lambda }{2}({N_1} + 1) \cdot {[0,1,2,\dots,{N_2} - 1]^{\rm T}},
	\end{align}
\end{subequations}
with $N_1$ and $N_2=N-N_1$ denoting the number of MAs in the inner and outer subarrays, respectively.

\subsubsection{Design Algorithm}
\begin{figure}[tp]
	\centering
	\includegraphics[width=0.9\linewidth]{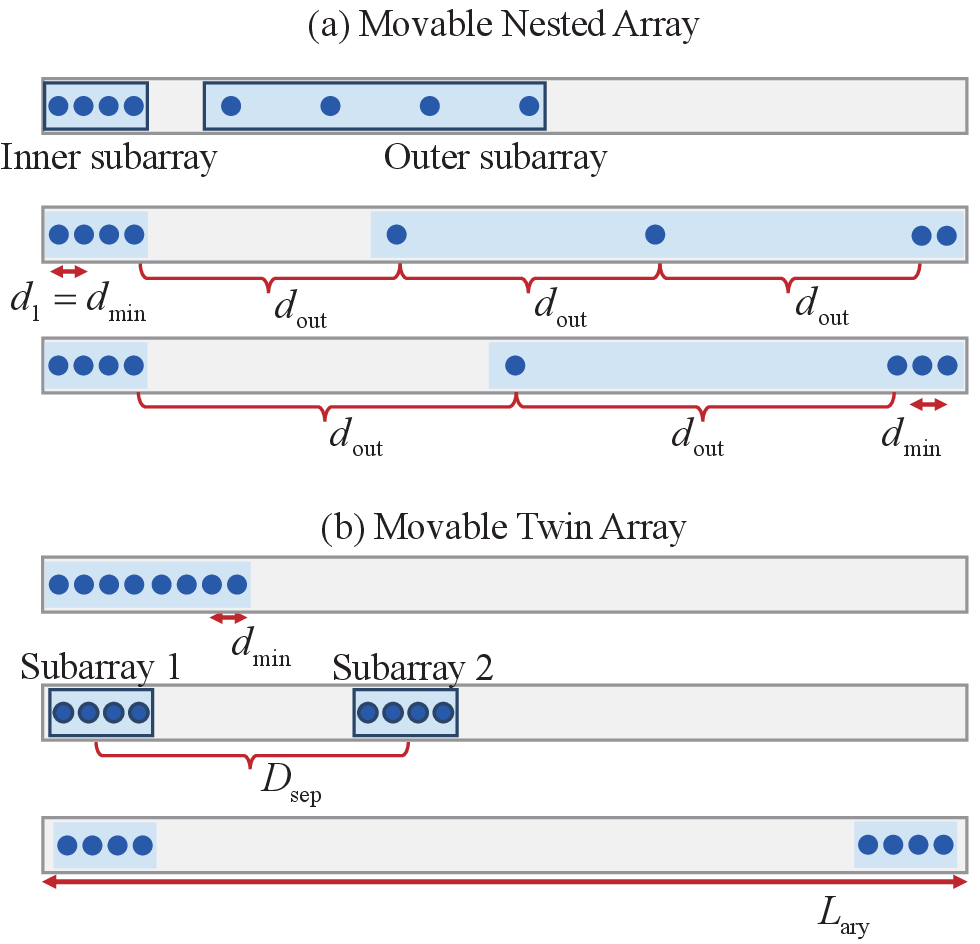}
	\caption{Illustration of the proposed (a) MNA and (b) MTA.}
	\label{fig_GNA_SPAconfig}	
\end{figure}
Building on the NA, we propose the MNA, as illustrated in Fig.~\ref{fig_GNA_SPAconfig}(a). 
As revealed in Remark~\ref{remark}, adapting to low- and high-SNR conditions requires the array geometry to balance the effective aperture against the sidelobe level.
To this end, the proposed MNA retains the partition $(N_1,N_2)$ of the NA, but fixes only the inner subarray ${\bf{x}}_{\rm{MN},1}$ at one aperture end with spacing $\lambda/2$, while the outer subarray ${\bf{x}}_{\rm{MN},2}$ is adjusted with SNR within the residual aperture.

For this structure, the number of distinct correlation lags is given by\cite{pal2010nested}
\begin{equation}
	L_{\rm NA,uni}=2(N-N_1)(N_1+1)-1.
	\label{lna_uni}
\end{equation}
Maximizing \eqref{lna_uni} yields the optimal inner-subarray size,
\begin{equation}
	N_1 =
	\begin{cases}
		N/2, & N \text{ is even},\\
		(N-1)/2, & N \text{ is odd}.
	\end{cases}
	\label{N1_opt}
\end{equation}

With $N_1$ given by \eqref{N1_opt}, the APV of MNA is
\begin{equation}
	{\bf{x}}_{\rm{MN}} = [{\bf{x}}_{\rm{MN},1}^{\rm T},{\bf{x}}_{\rm{MN},2}^{\rm T}]^{\rm T},
\end{equation}
where the inner subarray is fixed as
\begin{equation}
	{{\bf{x}}_{{\rm{MN}},{\rm{1}}}}= \frac{\lambda }{2} \cdot {[0,1,2,\ldots,{N_1} - 1]^{\rm T}},
	\label{A1}
\end{equation}
and the position of the $n$-th $(n\in[1,{N_2}])$ outer subarray MA is
\begin{equation}
	{{\bf{x}}_{{\rm{MN}},{\rm{2}}}}(n)\! =\! \min \left[ {\frac{\lambda }{2}{N_1} + n{d_{\rm{out}}},{L_{{\rm{ary}}}} \!-\! ({N_2} - n)\frac{\lambda }{2}} \right],
	\label{A2}
\end{equation}
in which the first term spaces the outer elements uniformly by $d_{\rm{out}}$ from the inner-subarray boundary, while the second term is triggered once an element reaches the array edge, causing it to aggregate there with the minimum spacing $\lambda/2$, as illustrated by the three configurations in Fig.~\ref{fig_GNA_SPAconfig}(a).

The optimal outer-subarray spacing $d_{\rm{out}}^*$ follows from a one-dimensional search minimizing the MSE approximation \eqref{MSE_ap}:
\begin{equation}
	\begin{split}
		d_{\rm{out}}^* &= \mathop{\arg\min}\limits_{d_{\rm{out}}} \mathrm{MSE}\left({\bf{x}}_{\rm{MN}}(d_{\rm{out}})\right)\\
		&\text{s.t.} \quad d_{\rm{out}} \in \left[\frac{\lambda}{2},\, L_{\rm{ary}}-N_1\frac{\lambda}{2}\right].
	\end{split}
	\label{d2_search}
\end{equation}

Unlike the traditional NA, whose outer spacing is fixed at $(N_1+1)\lambda/2$ and thus cannot adapt to the operating SNR, the MNA leverages antenna mobility to reconfigure this single parameter $d_{\rm{out}}^*$, continuously trading off the effective aperture against the sidelobe level.
At low SNR, $d_{\rm{out}}^*$ converges to the minimum spacing $\lambda/2$, reducing MNA to a compact ULA-like structure that suppresses the sidelobe level.
As SNR increases, $d_{\rm{out}}^*$ gradually expands to enlarge the effective aperture, until MNA converges to the SULA structure, the high-SNR-optimal geometry, at sufficiently high SNR. 

Through \eqref{d2_search}, the MNA spans the full range between the two limiting structures within the available moving region, which the fixed-spacing NA cannot achieve.

\subsection{Movable Twin Array}
The number of motors needed to driven MNA increases linearly with the MA number. 
To further reduce the hardware complexity, we revisit the two extreme optimal structures identified in Remark~\ref{remark}.
At high SNR, \cite{ma2024movable} has proven that the MSE-optimal APV is the SULA that maximizes the effective aperture, whereas at low SNR, dense ULA-like geometry is preferred to suppress the sidelobe level\cite{athley2007radar}.

Building on this result, we propose the MTA, as plotted in Fig.~\ref{fig_GNA_SPAconfig}(b), which divides the $N$ elements evenly into two half-wavelength-spaced subarrays and adjusts the inter-subarray spacing $D_{\rm{sep}}$.
The APV of MTA can be expressed as 
\begin{equation}
	{{\bf{x}}_{{\rm{MT}}}} = {[{{\bf{x}}_{{\rm{MT,1}}}}^{\rm T},{{\bf{x}}_{{\rm{MT,2}}}}^{\rm T}]^{\rm T}},
\end{equation}
in which
\begin{subequations}
	\label{MT_geo}
	\begin{align}
		&{{\bf{x}}_{{\rm{MT}},{\rm{1}}}} = \frac{\lambda }{2} \cdot {[0,1,2,\ldots,{N_1} - 1]^{\rm T}},\\
		&{{\bf{x}}_{{\rm{MT}},{\rm{2}}}} = {D_{\rm{sep}}} + \frac{\lambda }{2} \cdot {[0,1,2,\ldots,{N_2} - 1]^{\rm T}},
	\end{align}
\end{subequations}
where $N_{1} = \left\lfloor {N/2} \right\rfloor$, and the inter-subarray spacing ${D_{\rm{sep}}} \in \left[ {\frac{\lambda }{2}, {L_{{\rm{ary}}}} - \frac{\lambda }{2}({N_2} - 1)} \right]$.

The beam pattern of the MTA described in \eqref{MT_geo} can be regarded the coherent superposition of two identical ULA patterns with a relative phase shift.
As a special case of Theorem~\ref{Theorem_MSE}, the MSE of MTA can be derived in closed form as a function of $D_{\rm{sep}}$ as summarized below.


\begin{corollary}
\label{coro_sep}
For the MTA with inter-subarray spacing $D_{\rm{sep}}$ and APV given in \eqref{MT_geo}, the MSE of AoA estimation obtained from $T$ snapshots at a given SNR is approximated as
\begin{align}
	{\rm{MS}}{{\rm{E}}_{{\rm{MT}}}}
	\approx& P_0^{{\rm{MT}}} \left[ {\frac{{T\gamma }}{{{J_{{\rm{MT}}}}}}\Phi \left( {\sqrt {\frac{{T\gamma }}{2}} } \right) + {J_{{\rm{MT}}}}^{ - 1}{\Gamma _{3/2}}\left( {\frac{{T\gamma }}{4}} \right)} \right] \notag\\
	&+ \frac{{2{\lambda ^2}}}{{{D_{\rm{sep}}}^2}}\sum\limits_{k = 1}^{\left\lfloor {{D_{\rm{sep}}}/\lambda } \right\rfloor } {{k^2}P_k^{{\rm{MT}}}},
	\label{MSE_SP_coro}
\end{align}
where $\gamma  = \frac{{{N^2}{\rm{SN}}{{\rm{R}}^2}}}{{1 + N{\rm{SNR}}}}$, ${J_{{\rm{MT}}}}$ is the Fisher information given by 
\begin{equation}
	{J_{{\rm{MT}}}} = T\gamma  \cdot {\pi ^2}\left( {\frac{{{N^2}}}{{24}} + \frac{{2D_{\rm{sep}}^2}}{{{\lambda ^2}}}} \right). 
	\label{J_SP_coro}
\end{equation}
The probability of selecting the $k$-th lobe ($k=0$ for the mainlobe, $k=1,\ldots,\left\lfloor {D_{\rm{sep}}/\lambda } \right\rfloor $ for the sidelobes) is
\begin{equation}
	P_k^{{\rm{MT}}} = \Phi \left( {\frac{{(2k - 1)\lambda }}{{2{D_{\rm{sep}}}\sqrt {{B_{{\rm{MT}}}}} }}} \right) - \Phi \left( {\frac{{(2k + 1)\lambda }}{{2{D_{\rm{sep}}}\sqrt {{B_{{\rm{MT}}}}} }}} \right),
	\label{pk_SP}
\end{equation}
with ${B_{{\rm{MT}}}} = \dfrac{{48\left( {2 + N{\rm{SNR}}} \right)}}{{{\pi ^2}T{N^4}{\rm{SN}}{{\rm{R}}^2}}}$. 
\end{corollary}

The optimal inter-subarray spacing $D_{\rm{sep}}^*$ is determined by a one-dimensional search that minimizes the MTA MSE \eqref{MSE_SP_coro} subject to the aperture constraint, i.e., 
\begin{equation}
	\begin{split}
		{D_{\rm{sep}}}^* =\arg \mathop {\min }\limits_{{D_{\rm{sep}}}}  {\rm{MSE}}_{\rm{MT}}\left( {{{\bf{x}}_{{\rm{MT}}}}({D_{\rm{sep}}})} \right)&\\
		s.t. \quad {D_{\rm{sep}}} \in \left[ {\frac{\lambda }{2},{L_{{\rm{ary}}}} - \frac{\lambda }{2}({N_2} - 1)} \right].&
	\end{split}
	\label{Dx_search}
\end{equation}

\begin{remark}
The optimal spacing $D_{\rm{sep}}^*$ depends fundamentally on the SNR regime.
In the high-SNR regime, the MSE is dominated by the mainlobe error, so
\begin{align}
	\mathop {\lim }\limits_{{\rm{SNR}} \to \infty } {D_{{\rm{sep}}}}^* 
	&\triangleq \mathop {\lim }\limits_{{\rm{SNR}} \to \infty } \mathop {\arg \min }\limits_{{D_{{\rm{sep}}}}} {\rm{MS}}{{\rm{E}}_{{\rm{MT}}}}\notag\\
	&=\mathop {\arg \max }\limits_{{D_{{\rm{sep}}}}} {J_{{\rm{MT}}}}. 
\end{align}
Since $J_{{\rm{MT}}}$ in \eqref{J_SP_coro} increases monotonically with $D_{\rm{sep}}$, we have
\begin{align}
	\mathop {\lim }\limits_{{\rm{SNR}} \to \infty } {D_{{\rm{sep}}}}^* 
	&= \max \left( {{D_{{\rm{sep}}}}} \right) \notag\\
	&= {L_{{\rm{ary}}}} - \frac{\lambda }{2}({N_2} - 1),
	\label{Dx_highSNR}
\end{align}
indicating that at high SNR, the two subarrays are positioned at both ends of the region to maximize the effective aperture\cite{ma2024movable}.

As SNR decreases, ${\rm{MSE}}_{{\rm{MT}}}$ is dominated by the mainlobe error ($P_0^{\rm{MT}}$) and the first sidelobe ambiguity ($P_1^{\rm{MT}}$)\cite{bell1996explicit}, and \eqref{MSE_SP_coro} reduces to 
\begin{align}
	&\mathop {{\rm{lim}}}\limits_{{\rm{SNR}} \to 0} {\rm{MSE}}_{{\rm{MT}}} 
	\overset{(a)}{\approx} P_0^{{\rm{MT}}}\frac{{T\gamma }}{{2{J_{{\rm{MT}}}}}} + \frac{{2{\lambda ^2}}}{{{D_{\rm{sep}}}^2}}P_1^{{\rm{MT}}}\notag\\
	&\overset{(b)}{\approx}\frac{{12}}{{{\pi ^2}\left( {{N^2} + \frac{{48D_{\rm{sep}}^2}}{{{\lambda ^2}}}} \right)}} + \sqrt {\frac{{8{B_{{\rm{MT}}}}}}{\pi }} \frac{\lambda }{{{D_{\rm{sep}}}}}{e^{ - \frac{{{\lambda ^2}}}{{8{D_{\rm{sep}}}^2{B_{{\rm{MT}}}}}}}},
	\label{MSE_L}
\end{align}
where $(a)$ applies $\Phi(0)=\tfrac12$, $\Gamma_{3/2}(0)=0$, and truncates the sidelobe sum to its dominant $k=1$ term; $(b)$ approximates $P_0^{\rm{MT}}\approx1$ (valid since $\Phi\to\tfrac12$ makes $P_0^{\rm{MT}}\to0$), substitutes \eqref{J_SP_coro} into the resulting expression, and applies $\Phi(z) \approx \frac{1}{\sqrt{2\pi}z} e^{-z^2/2}$ \cite{athley2000direction} to $P_1^{\rm{MT}}$, with $e^{-9x}\ll e^{-x}$ for $x>0$ used to retain only the dominant exponential.

When $\mathrm{SNR}\to0$, the exponential term in \eqref{MSE_L} dominates, i.e., ${\rm{MSE}}_{\rm{MT}} \propto \exp({ - \frac{{{\lambda ^2}}}{{8{D_{\rm{sep}}}^2{B_{{\rm{MT}}}}}} })$, which gives
\begin{align}
	\mathop {{\rm{lim}}}\limits_{{\rm{SNR}} \to 0} {D_{\rm{sep}}}^*
	&\approx \mathop {{\rm{lim}}}\limits_{{\rm{SNR}} \to 0} \mathop {\arg \min }\limits_{{D_{{\rm{sep}}}}} {e^{ - \frac{{{\lambda ^2}}}{{8{D_{{\rm{sep}}}}^2{B_{{\rm{MT}}}}}}}}\notag\\
	&= \min( {D_{\rm{sep}}}) = \frac{\lambda }{2}.
\end{align}
That is, at low SNR, the two subarrays of MTA merge into a dense $\lambda/2$-spaced ULA to suppress the sidelobe level\cite{athley2003space}.
\end{remark}

\begin{table}[]
	\caption{Comparison of computational and hardware complexities among different array design schemes.}
	\label{Table_cmp}
	\renewcommand{\arraystretch}{1}
	\resizebox{\columnwidth}{!}{%
		\begin{tabular}{l|c|c}
			\hline
			\textbf{Array Geometry}  & \textbf{Computational Complexity}              & \textbf{Motors} \\ \hline
			Exhaustive Search        & ${\cal O}\!\left[e^{N_{\rm{all}} H\left(\frac{N}{N_{\rm{all}}}\right)}\right]$ & $N$             \\
			MFA & ${\cal O}\!\left(N_{\rm{all}}^{2}\right)$ & $N$ \\
			MNA & ${\cal O}\!\left(\dfrac{L_{\rm{ary}}-\frac{\lambda}{2}N_1}{\Delta d}\right)$          & $N/2$           \\ 
			MTA & ${\cal O}\!\left(\dfrac{L_{\rm{ary}}-\frac{\lambda}{2}N_2}{\Delta d}\right)$          & $1$             \\ \hline 
		\end{tabular}%
	}
\end{table}

\subsection{Computational Complexity}
Next, we compare the computational complexities of proposed MFA, MNA and MTA, detailed as follows and summarized in Table \ref{Table_cmp}.

\subsubsection{Exhaustive Search}
By contrast, exhaustive search must evaluate all possible APVs by selecting $N$ antenna positions from $N_{\rm{all}}={\left\lfloor {{L_{{\rm{ary}}}}/\Delta d} \right\rfloor }$ candidates, giving a complexity of
\begin{equation}
	\begin{split}
		{\cal O}({\rm{Exhaust}}) 
		&=\binom{N_{\rm{all}}}{N}\\
		&= \frac{{{N_{{\rm{all}}}}!}}{{N!({N_{{\rm{all}}}} - N)!}}.
	\end{split}
	\label{cmp_Exhaust}
\end{equation}
Applying Stirling's approximation $\ln n! \approx n\ln n - n$ to \eqref{cmp_Exhaust}, it gives
\begin{equation}
	{\cal O}({\rm{Exhaust}}) \approx {\cal O}\!\left[e^{N_{\rm{all}} H\left(\frac{N}{N_{\rm{all}}}\right)}\right],
	\label{cmp_exhaustive}
\end{equation}
where $H(x)\triangleq-x\ln x-(1-x)\ln(1-x)$.
This result indicates that the complexity of exhaustive search increases exponentially with the number of candidate positions $N_{\rm all}$, which is determined by the array aperture $L_{\rm ary}$ and the position resolution $\Delta d$. Consequently, exhaustive search becomes computationally prohibitive for large-aperture arrays.

\subsubsection{MFA}
The computational complexity of Algorithm~\ref{alg_greedy} consists of two parts.  
The complexity of \textbf{Stage~1} is
\begin{align}
	{\cal O}(St1) 
	&= \sum\limits_{i = 0}^{{N_{{\rm{all}}}} - N - 1} {\left( {{N_{{\rm{all}}}} - i} \right)} \notag\\
	&= \frac{{{N_{{\rm{all}}}}({N_{{\rm{all}}}} + 1) - N(N + 1)}}{2}\notag\\
	&\approx {\cal O}({N_{{\rm{all}}}}^2),
	\label{cmp_st1}
\end{align}
which grows quadratically with the available positions ${N_{\rm all}}$.

The complexity of \textbf{Stage~2} is given by
\begin{align}
	{\cal O}(St2) 
	&= \sum\limits_{i = 1}^{{I_{{\rm{itr}}}}} {\sum\limits_{n = 1}^N {\left| {{\cal C}_n^i} \right|} } \notag\\
	&\le {I_{{\rm{itr}}}}N\left( {{N_{\rm{all}}} - N\frac{\lambda }{{2\Delta d}}} \right)\notag\\
	&= {\cal O}({I_{{\rm{itr}}}}N{N_{{\rm{all}}}}),
	\label{cmp_st2}
\end{align}
where $I_{\rm itr}$ denotes the number of refinement iterations. Since the refinement procedure rapidly converges, $I_{\rm itr}$ remains small in practice and is typically less than 10.

Combining \eqref{cmp_st1} and \eqref{cmp_st2}, the computational complexity of Algorithm~\ref{alg_greedy}  can be approximated as 
\begin{align}
	{\cal O}(\rm{MFA}) 
	&= {\cal O}(St1) + {\cal O}(St2) \notag\\
	&= {\cal O}({N_{{\rm{all}}}}^2 + {I_{{\rm{itr}}}}N{N_{{\rm{all}}}})\notag\\
	&\approx  {\cal O}({N_{{\rm{all}}}}^2)={\left( {\frac{{{L_{{\rm{ary}}}}}}{{\Delta d}}} \right)^2},\label{cmp_A1}
\end{align}
where the approximation holds since $I_{\rm itr} \ll N_{\rm all}$ in practice.

Compared with the exhaustive search, whose complexity in \eqref{cmp_exhaustive} increases exponentially with $N_{\rm all}$, MFA is realized with a reduced computational complexity of a polynomial order.


\subsubsection{MNA}
The computational complexity of the MNA scheme primarily stems from the one-dimensional search over the outer-element spacing $d_{\rm{out}}$ in \eqref{d2_search}.
When the search step is set to $\Delta d$, its computational cost can be expressed as
\begin{equation}
	\begin{split}
		{\cal O}({\rm{MN}}) 
		&= \frac{{{L_{{\rm{ary}}}} - \frac{\lambda }{2}{N_1} - \frac{\lambda }{2}}}{{\Delta d}}\\
		&={\cal O}\left( {\frac{{{L_{{\rm{ary}}}}-\frac{\lambda}{2}N_1}}{{\Delta d}}} \right),
	\end{split}	
	\label{cmp_MNA}
\end{equation}
which shows that the complexity of MNA scales linearly with the residual aperture $L_{\rm{ary}}-\frac{\lambda}{2}N_1$ available to the outer subarray, and inversely with the position resolution $\Delta d$.

\subsubsection{MTA}
The computational complexity of MTA arises from a similar one-dimensional search over $D_{\rm{sep}}$ in \eqref{Dx_search}.
With the position resolution $\Delta d$, the complexity is
\begin{equation}
	\begin{split}
		{\cal O}({\rm{MT}}) 
		&=\frac{{{L_{{\rm{ary}}}} - \frac{\lambda }{2}{N_2}}}{{\Delta d}}\\
		&={\cal O}\left( {\frac{{{L_{{\rm{ary}}}}-\frac{\lambda}{2}N_2}}{{\Delta d}}} \right),
	\end{split}
\end{equation}
which scales linearly with the residual aperture $L_{\rm{ary}}-\frac{\lambda}{2}N_2$ left for the outer subarray, of the same order as MNA in \eqref{cmp_MNA}.

\subsection{Hardware Complexity}
The three MA array structures differ in the number of motors required for MA adjustments, as summarized in Table~\ref{Table_cmp}. 
The exhaustive search and MFA require $N$ independent motors to provide full antenna mobility, offering the largest configuration freedom at the highest hardware complexity.
Between them, the MFA is preferable, reducing the computational complexity from the exponential order of exhaustive search (see eq.~\eqref{cmp_exhaustive}) to a linear order (see eq.~\eqref{cmp_A1}).

By fixing the inner subarray and optimizing only the outer antennas, MNA reduces the required motors to $N/2$ while preserving sufficient array reconfigurability, as illustrated in Fig.~\ref{fig_GNA_SPAconfig} (a).  
As an extreme case with minimum hardware cost, the MTA requires only one motor to adjust the spacing between two fixed subarrays, as shown in Fig.~\ref{fig_GNA_SPAconfig} (b).
At the cost of reduced configuration flexibility, the MTA is well suited to large arrays under stringent hardware constraints.

\section{Numerical Results}

In this section, we first validate the accuracy of the derived MSE approximation via Monte Carlo (MC) simulations with $2000$ trials and $T=1000$ snapshots, benchmarked against the CRB, the EZZB~\cite{khan2013explicit}, and the existing approximation in~\cite{athley2005threshold}.
We then evaluate the proposed MFA, MNA, and MTA against the exhaustive-search global optimum (labelled Exhaustive), the state-of-the-art MA design~\cite{ma2026movable} (labelled SCA), the CRB-optimized SULA~\cite{ma2024movable} (labelled SULA), and two representative FPAs, namely ULA and MRA~\cite{moffet1968minimum}, where the latter adopts a classic sparse geometry designed to suppress sidelobes.
Unless otherwise specified, the array region is $L_{\rm ary}=30\lambda$ with $N=8$ movable antenna elements, the minimum inter-antenna spacing is $d_{\min} = \lambda/2$, position resolution is set to $\Delta d=0.25\lambda$, and the target AoA $\theta = 30^\circ$.
The tested SNR ranges from ${\rm{SNR}}\in[-20,5]$ dB.

\subsection{Effectiveness of proposed MSE Approximation and Three MA Structures}

\begin{figure}[tp]
	\centering
	\includegraphics[width=\linewidth]{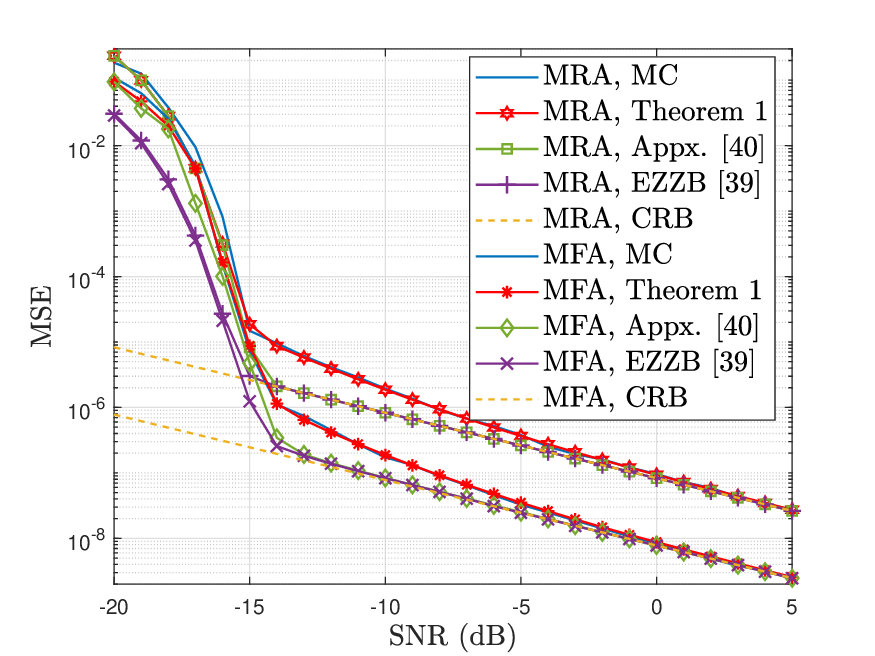}
	\caption{Validation of the proposed MSE approximation for the MA-based MFA and the MRA-based FPA under different SNR.}
	\label{fig_bound_cmp}	
\end{figure}

Fig.~\ref{fig_bound_cmp} validates the derived MSE approximation against the MC simulation, the existing approximation \cite{athley2005threshold}, the EZZB \cite{khan2009explicit}, and the CRB, under $L_{\rm{ary}}=30\lambda$, using two distinct APVs, i.e., the proposed MFA and the classic FPA-based MRA \cite{moffet1968minimum}, to verify its generality.
For both APVs, the derived expression (\textit{Proposed}) matches the MC results more closely than the existing approximation and the EZZB, accurately tracking the MSE drop over ${\rm{SNR}}\in[-18,0]$~dB, whereas the CRB fails throughout the low-to-medium SNR regime, confirming that the derived approximation effectively captures the sidelobe-induced error.
Note that the effectiveness of the derived approximation degrades when ${\rm{SNR}}<-18$~dB, where ${\rm{MSE}}\geq 10^{-2}$ already falls outside the practical sensing requirement, so its applicability for APV optimization is unaffected. 
Moreover, the MFA outperforms the MRA by one-to-two orders of magnitude for ${\rm SNR}>-15$~dB, confirming that fully exploiting the movable region via SNR-adaptive APV reconfiguration is more effective across a wide SNR range than relying on a fixed geometry.

\begin{figure}[tp]
	\centering
	\includegraphics[width=\linewidth]{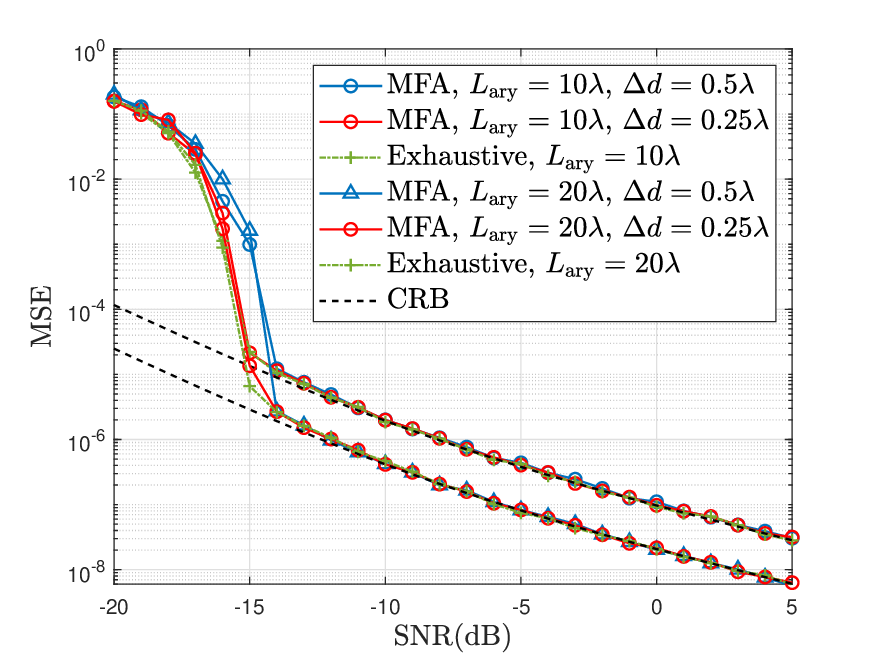}
	\caption{ Validation of the MFA against exhaustive search under different array region sizes $L_{\mathrm{ary}}$ and step sizes $\Delta d$.}
	\label{fig_stepsize_cmpExhaustive}	
\end{figure}

As illustrated in Fig. \ref{fig_stepsize_cmpExhaustive}, the MFA converges asymptotically to the exhaustive-search performance as the position grid $\Delta d$ decreases. 
In particular, using a resolution of $\Delta d=0.25\lambda$ is sufficient to attain near-exhaustive accuracy with negligible residual gaps over the tested SNR range.
Since an excessively position resolution would only incur unnecessary computational complexity, we therefore adopt $\Delta d=0.25\lambda$ in the following simulations. 
Moreover, enlarging the array region size from $L_{\mathrm{ary}}=10\lambda$ to $20\lambda$ yields little benefit in the low-SNR regime, where all curves behave almost identically for ${\rm SNR}\leq -15$ dB. 
The reason is that a sparer array geometry sharpens the mainlobe but simultaneously increases the sidelobe level, which increases the estimation error at low SNR. 
This confirms that adapting the APV to the SNR is more critical than merely enlarging the effective aperture for low-to-medium SNR sensing.


\begin{figure}[tp]
	\centering
	\includegraphics[width=\linewidth]{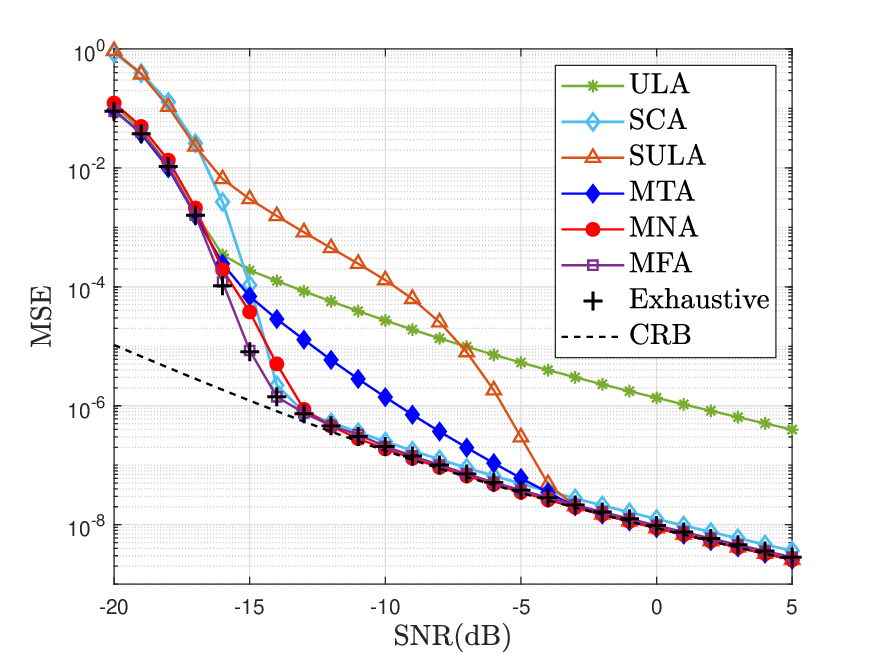}
	\caption{MSE performance of the proposed MFA, MNA, and MTA compared with MA-based SCA\cite{ma2026movable} and classic FPAs under $N=8,L_{\rm{ary}}=30\lambda$.}
	\label{fig_fixcmp_MSEap_lary30_M8_deltad0.5_noMT}
\end{figure}


Fig.~\ref{fig_fixcmp_MSEap_lary30_M8_deltad0.5_noMT} compares the proposed schemes with existing MA and representative FPA designs under $N=8$.
First, compared with the exhaustive-search global optimum, the proposed MFA achieves near-optimal performance over the tested SNR range, while requiring fewer motors and much lower computational complexity.
Besides, optimized based on the derived MSE approximation, the MFA achieves about one order-of-magnitude lower MSE over the state-of-the-art sidelobe-aware MA benchmark SCA \cite{ma2026movable} at low SNR regime, where SCA still suffers from degradation similar to SULA.
Second, the proposed MNA closely approaches the optimal Exhaustive curve with only $N/2$ motors, compared to the MFA required full $N$ motors.
The reduced DoF lead to a noticeable loss only within the narrow range of ${\rm SNR}\in[-16,-14]$~dB, and thereafter overlapping with the optimal performance.
Third, using a single motor, the MTA tracks the global optimum at both low and high SNR, with one order of magnitude MSE increment within ${\rm SNR}\in[-15,-5]$~dB, where its performance is limited by its mobility.
Nevertheless, benchmarked against the two fixed FPA baselines, the MTA still outperforms the ULA, which is favoured at low SNR yet cannot exploit the full aperture, by about one-to-two orders of magnitude for ${\rm SNR}>-10$~dB, and outperforms the SULA, the high-SNR-optimal, by up to two orders of magnitude for ${\rm SNR}<-5$~dB by adjusting the two subarrays to suppress sidelobe level.

\begin{figure}[tp]
	\centering
	\includegraphics[width=\linewidth]{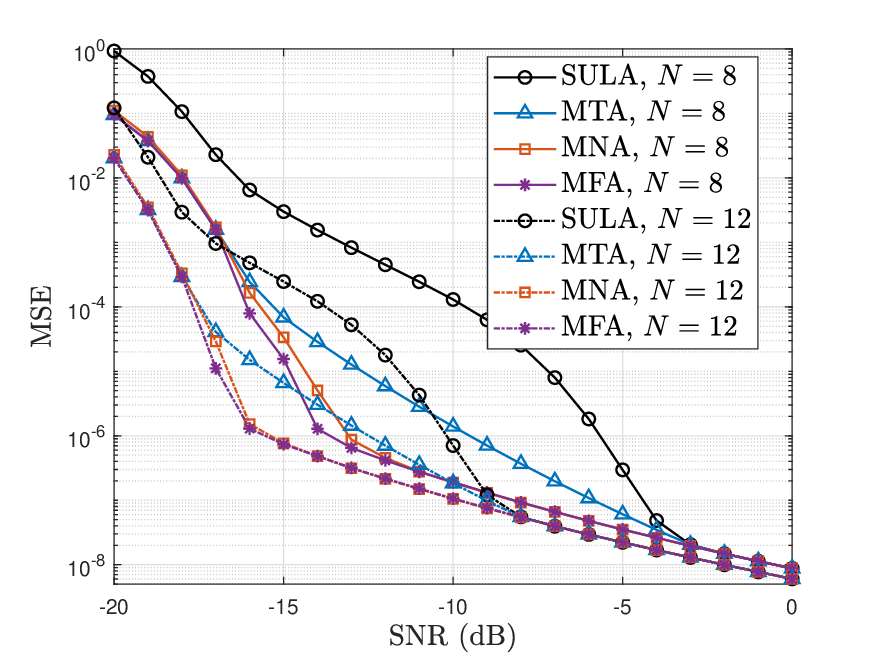}
	\caption{MSE performance comparison of the MNA and MTA structures under different numbers of MAs $N$.}
	\label{fig_nesd_sep_cmp_M8_M12}	
\end{figure}


Fig.~\ref{fig_nesd_sep_cmp_M8_M12} compares the MSE of the proposed structures under $N=8$ and $N=12$ MAs.
The $N=8$ and $N=12$ curves of each architecture nearly coincide when $\mathrm{SNR}<-15$~dB, separate visibly within $\mathrm{SNR}\in[-15,-5]$~dB, and merge again near the CRB floor at high SNR, reflecting a shift in the limiting factor across these regimes.
At low SNR, all three proposed structures converge to a similar sidelobe-suppressing, ULA-like structure, so the MSE is primarily constrained by the antenna number $N$.
Within $\mathrm{SNR}\in[-15,-5]$~dB, the structures diverge according to their motor freedom.
Specifically, the single-motor MTA degrades the most, while the MNA ($N/2$ motors) and the fully flexible MFA ($N$ motors) achieve up to one order of magnitude lower MSE.
This gap between MNA and the MFA narrows as $N$ increases to $12$, indicating that a larger antenna number can partially compensate for the limited array mobility.
By contrast, for $\mathrm{SNR}<-5$~dB, the fixed SULA with $N=12$ remains one to two orders of magnitude worse than the proposed MNA with only $N=8$, showing that increasing the antenna number cannot compensate for a fixed, SNR-mismatched geometry, highlighting the importance of SNR-adaptive array reconfiguration.


In summery, Figs.~\ref{fig_bound_cmp}-\ref{fig_nesd_sep_cmp_M8_M12} confirm the accuracy of the derived MSE approximation and the effectiveness of the proposed MA structures.
Notably, the MNA closely approaches the globally optimal exhaustive search with only $N/2$ motors, while the MTA, despite using a single motor, achieves one-order-of-magnitude MSE reduction over FPAs such as SULA.


\subsection{Analysis of Proposed Geometries}

\begin{figure}[tp]
	\centering
	\includegraphics[width=\linewidth]{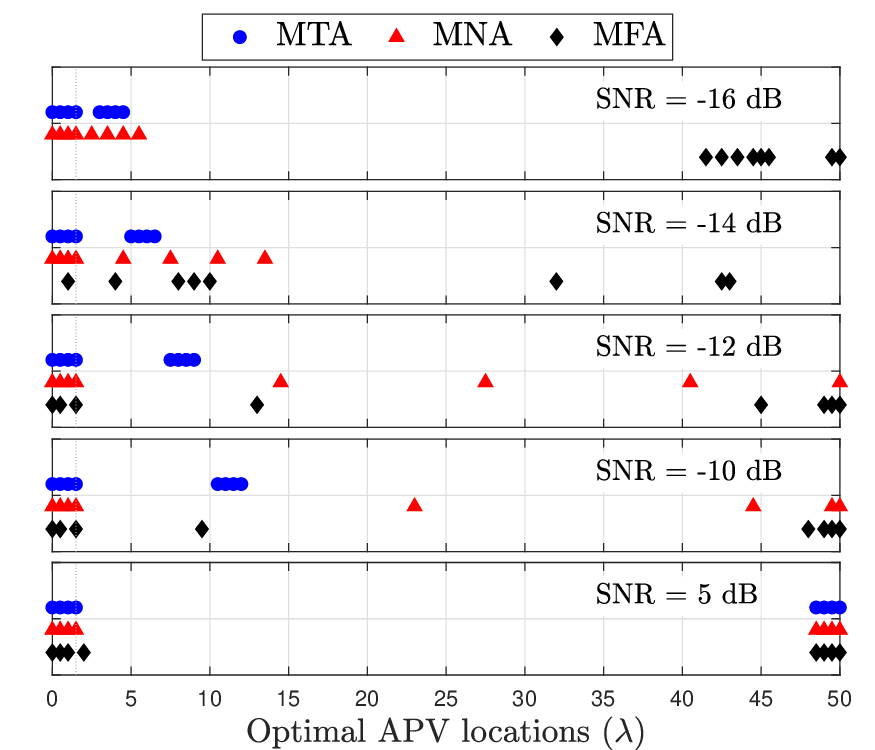}
	\caption{Optimal APVs of the proposed MTA, MNA and MFA under different SNR conditions.}
	\label{fig_NestedSep_APV}	
\end{figure}

Fig.~\ref{fig_NestedSep_APV} illustrates the optimal APVs of the MTA, MNA, and MFA under $L_{\rm{ary}}=50\lambda$ and $N=8$, highlighting their distinct geometric evolution across SNR.
The MTA and MNA follow structured geometries using only a single motor and $N/2$ motors, respectively, whereas the MFA requires $N$ fully independent motors.
At low SNR ($\mathrm{SNR}=-16$~dB), the MTA and MNA merge their two subarrays at the origin into a ULA-like geometry, while the MFA converges to a similarly compact half-wavelength geometry located near the opposite aperture boundary.
As SNR increases, the subarray spacing of the MTA grows monotonically, while the outer subarray of the MNA progressively expands toward the aperture edge to enlarge the effective aperture, and once the outer elements reach the boundary, aggregates there with half-wavelength spacing.
In contrast, the APV of MFA exhibits no consistent pattern across SNR, since its two-stage pruning may settle into different local optima at each SNR.
At high SNR, all three structures form two dense clusters at the opposite aperture ends, converging to a SULA-like structure to maximize the effective aperture.
This evolution reflects a trade-off between enlarging the array aperture and suppressing sidelobe level as SNR varies, which the MTA and MNA resolve through a structured, motor-efficient adjustment, while the MFA attains the same trade-off but requires full antenna mobility.

\begin{figure}[tp]
	\centering
	\includegraphics[width=\linewidth]{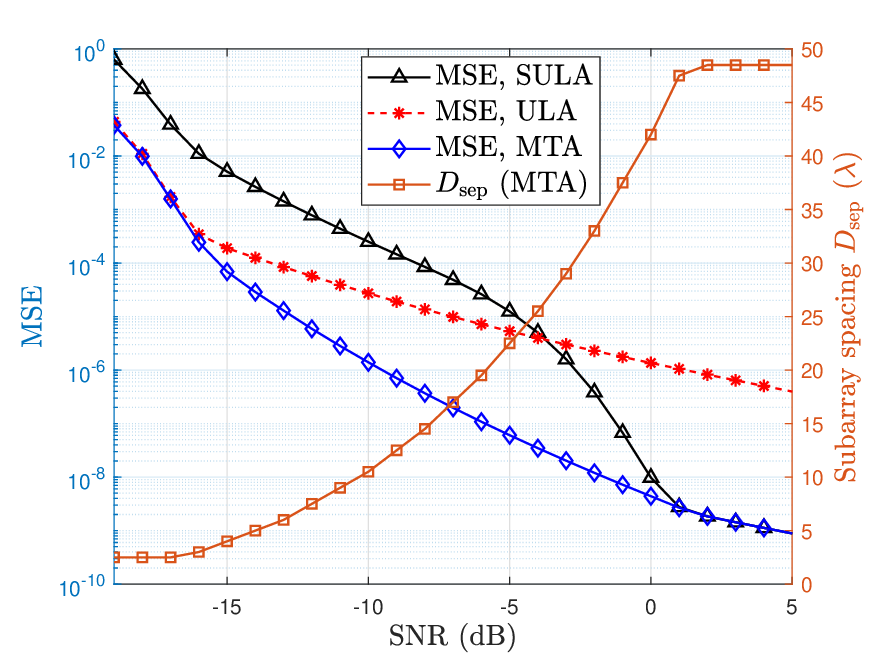}
	\caption{Optimal subarray spacing of MTA versus SNR.}
	\label{fig_dsepopt_cmpMSE}	
\end{figure}

Fig.~\ref{fig_dsepopt_cmpMSE} quantifies the evolution of the optimal inter-subarray spacing $D_{\rm{sep}}^{*}$ of MTA, and the resulting MSE.
For ${\rm{SNR}}<-16$~dB, $D_{\rm{sep}}^{*}$ remains near its minimum $\lambda/2$, so the MTA reduces to a dense ULA-like geometry that suppresses sidelobe ambiguity, appears as the overlapped curve with ULA.
As SNR increases, $D_{\rm{sep}}^{*}$ grows continuously, slowly at first and then more steeply beyond about ${\rm{SNR}}\in[-5,1]$~dB, progressively enlarging the effective aperture and gradual increasing the MSE gap to two-orders of magnitude compared to fixed ULA.
As ${\rm{SNR}}>1$~dB, $D_{\rm{sep}}^{*}\to\lambda/2$, where the two subarrays reach the opposite ends of the aperture, and the MTA curve merges with the SULA benchmark, both converging toward the CRB limited by $N$ and the region size.
This close correspondence between the continuous growth of $D_{\rm{sep}}^{*}$ and the steady MSE reduction confirms that the performance gain of the MTA arises from continuously adapting its structure to the SNR, rather than from a single fixed advantageous geometry.


\begin{figure}[tp]
	\centering
	\includegraphics[width=\linewidth]{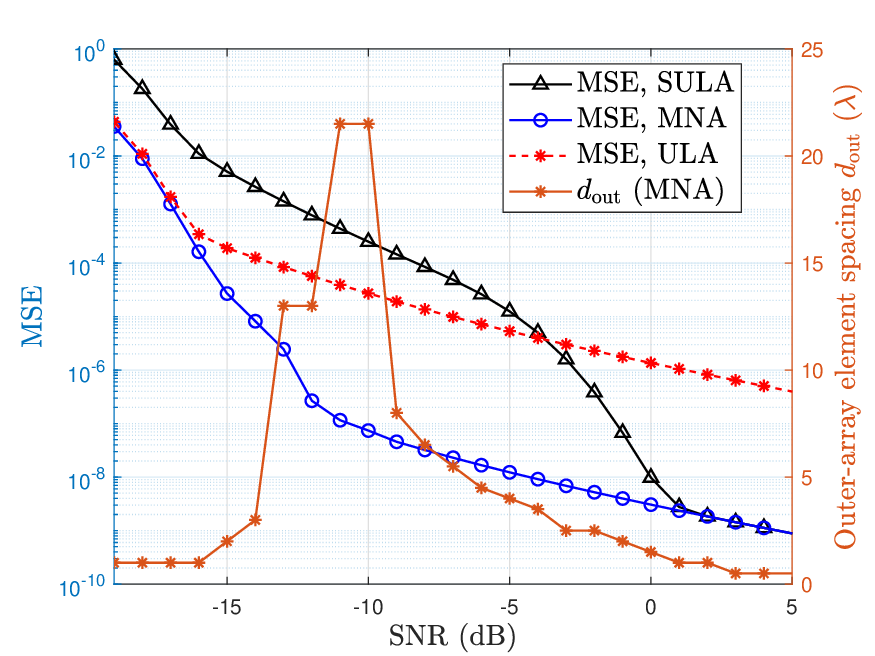}
	\caption{Optimal outer-element spacing of MNA versus SNR.}
	\label{fig_dnested}	
\end{figure}

Fig.~\ref{fig_dnested} quantifies the evolution of $d_{\rm{out}}^{*}$ to validate the non-monotonic trend observed in Fig.~\ref{fig_NestedSep_APV}.
For $\mathrm{SNR}<-16$~dB, $d_{\rm{out}}^{*}$ remains at its minimum $\lambda/2$, so the MNA reduces to a ULA-like geometry whose MSE closely tracks the fixed ULA curve.
Within $\mathrm{SNR}\in[-14,-10]$~dB, $d_{\rm{out}}^{*}$ rises sharply to enlarge the effective aperture until the outer elements reach the boundary of the movable region, yielding a pronounced MSE drop.
Beyond this point, the aperture constraint in \eqref{A2} forces the outer elements to progressively aggregate at the array edge, which appears as a decrease in $d_{\rm{out}}^{*}$.
As SNR increases further, $d_{\rm{out}}^{*}\to\lambda/2$,  where the outer elements settle near the boundary with spacing, and the MNA curve merges with the SULA curve, confirming its convergence to the SULA geometry at high SNR.
Overall, this evolution allows the MNA to inherit the sidelobe robustness of a dense ULA at low SNR while fully exploiting the available aperture at high SNR, yielding a one-to-three-order-of-magnitude MSE reduction over the fixed SULA in the low-to-medium SNR regime, and about one order of magnitude over the fixed ULA at high SNR.

\begin{figure}[tp]
	\centering
	\includegraphics[width=\linewidth]{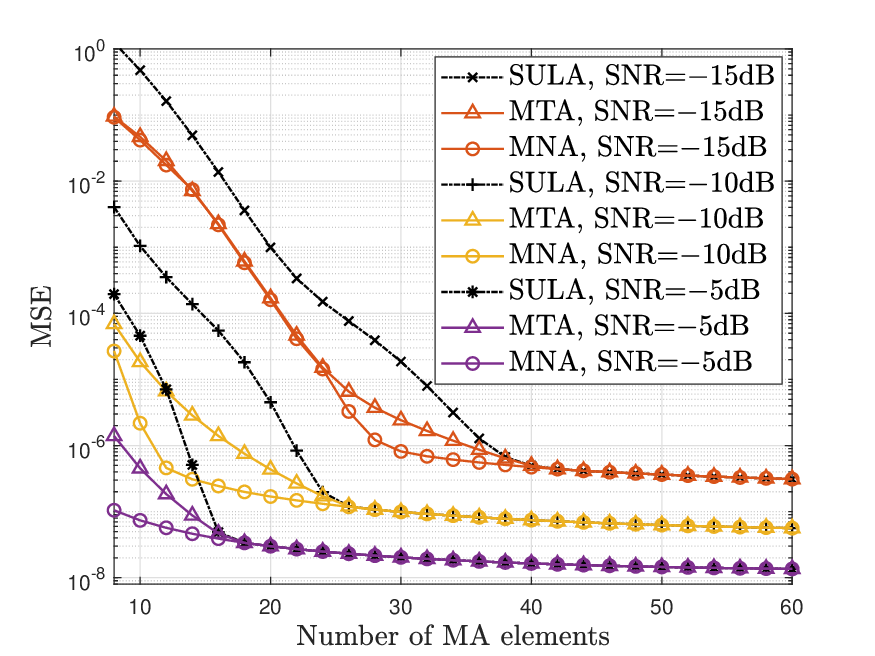}
	\caption{MSE performance versus the number of MA elements.}
	\label{fig_Mcmp_nested_sep}	
\end{figure}

Fig.~\ref{fig_Mcmp_nested_sep} compares the MSE of the proposed MNA, MTA, and the fixed SULA versus the number of MA elements $N$ under $L_{\rm{ary}}=30\lambda$.
When $N$ is limited, i.e., $N<40$ for ${\rm SNR}=-15$~dB, $N<25$ for ${\rm SNR}=-10$~dB, and $N<15$ for ${\rm SNR}=-5$~dB, both MNA and MTA substantially outperform the fixed SULA, the high-SNR-optimal geometry~\cite{ma2024movable}, with the gap widening as $N$ decreases.
As $N$ increases beyond these thresholds, all three structures progressively converge to a common MSE floor.
Owing to the additional degree of freedom offered by its $N/2$ movable outer elements, MNA attains up to a two-order-of-magnitude MSE reduction over SULA, exceeding the one-order-of-magnitude reduction achieved by the single-motor MTA in the same element-limited regime.
Specifically, at ${\rm SNR}=-10$~dB, MNA requires only about $1/4$, while MTA requires about $1/2$, of the elements needed by SULA to reach the same MSE.
These results indicate that, in the element-limited and low-to-medium SNR regime, SNR-adaptive APV reconfiguration is more effective for improving estimation accuracy than simply increasing the antenna number, with MNA offering the best accuracy-hardware trade-off among the proposed structures owing to its extra motor freedom.


Overall, Figs.~\ref{fig_NestedSep_APV}-\ref{fig_Mcmp_nested_sep} confirm that, in the low-to-medium SNR regime, SNR-adaptive MAs are more effective than simply increasing the number of fixed-position antennas.
Among the proposed schemes, the MNA closely approaches the globally optimal geometry with a lower-complexity structured design, while the MTA, at the lowest hardware cost, still delivers a substantial accuracy gain with a single motor.


\section{Conclusion}

This paper proposed adaptive MA structures to enhance AoA estimation accuracy across the low-to-medium SNR regime.
We first derived a deterministic MSE approximation that captures the local estimation error within the mainlobe and the outlier error induced by ambiguity-function sidelobes, thereby overcoming the limitation of the CRB in low-SNR regime.
We then developed three adaptive MA structures that trade off hardware and computational complexities. 
For systems with fully independent motors, the MFA is optimized by a two-stage greedy search, which approaches the globally optimal exhaustive search at lower computational complexity. 
By fixing the dense inner subarray and adjusting only the outer-subarray spacing, the MNA halves the required motors while closely approaching the globally optimal performance.
Tuning only the inter-subarray spacing with a single motor, the MTA further reduces the hardware complexity, while still achieving an order-of-magnitude MSE reduction over the fixed SULA geometry.
Numerical results confirm that in the low-to-medium SNR regime, SNR-adaptive APV improves estimation accuracy more effectively than simply increasing the antenna number.
The proposed structures offer hardware- and computationally-efficient solutions for robust AoA estimation across a wide SNR range.
Future work will extend to multi-target sensing and two-dimensional MA design. 


\begin{appendices}	
\section{}\label{appendix_MSE}
Following the MSE decomposition in \eqref{MSE_split}, this appendix derives the sidelobe-selection probability $P_n$ and the mainlobe estimation error $\mathbb{E}[\epsilon_u^2]$.
Selecting the $n$-th sidelobe $u_n$ over the mainlobe $u_0$ is equivalent to a binary hypothesis test between the two, whose probability $P_n\triangleq\Pr[V(u_n)>V(u_0)]$ follows from \cite{proakis2001digital} as 
\begin{align}
	{P_n} &= Q\left( {\sqrt {\frac{S}{2}\left( {1 - \sqrt {1 - r_n^2} } \right)} ,\sqrt {\frac{S}{2}\left( {1 + \sqrt {1 - r_n^2} } \right)} } \right)\notag\\
	&- {e^{ - \frac{S}{2}}}{I_0}\left( {\frac{{{r_n}S}}{2}} \right)\left[ {1 - \frac{1}{{{2^{2T - 1}}}}\sum\limits_{m = 0}^{T -1}{\binom{2T-1}{m}} } \right]\notag\\
	&+ \frac{1}{{{2^{2T - 1}}}}\sum\limits_{m = 1}^{T - 1} {\left\{ {{e^{ - \frac{S}{2}}}{I_m}\left( {\frac{{{r_n}S}}{2}} \right)\sum\limits_{k = 0}^{2T - 1 - m}\binom{2T-1}{k} } \right.}\notag \\
	&\left. { \times \left[ {{{\left( {\frac{{1 + \sqrt {1 - r_n^2} }}{{{r_n}}}} \right)}^m} - {{\left( {\frac{{1 - \sqrt {1 - r_n^2} }}{{{r_n}}}} \right)}^m}} \right]} \right\},
	\label{Pn_expression}
\end{align}
where $S\triangleq NT\cdot{\rm{SNR}}$ is the effective SNR accumulated over $N$ antennas and $T$ snapshots, $r_n\triangleq \frac{1}{N}{\left| {{{\bf{\alpha }}^H}({\bf{x}},{u_n}){\bf{\alpha }}({\bf{x}},{u_0})} \right|}\in[0,1]$ is the ambiguity function at the $n$-th sidelobe location $u_n$ and the true angle $u_0$, $I_m(x)$ is the order-$m$ modified Bessel function of the first kind, 
\begin{equation}
	I_m(x) = \sum_{k=0}^{\infty} \frac{(x/2)^{2k+m}}{k!\Gamma(m+k+1)},
	\label{bessel}
\end{equation}
and $Q(\alpha, \beta)$ is Marcum's $Q$ function,
\begin{equation}
	Q(\alpha, \beta) = \int_{\beta}^{\infty} t e^{-(t^2+\alpha^2)/2} I_0(\alpha t)\, dt.
	\label{MarcumQ}
\end{equation}

Next, we derive the mainlobe error $\mathbb{E}[\epsilon_u^2]$ in \eqref{MSE_split}. 
Consider a random $u_0$ with a priori pdf $p(u)$, $-1\le u\le1$.
Relating MSE as a binary detection problem, $\mathbb{E}[\epsilon_u^2]$ is lower-bounded by the minimum error probability of the optimal binary hypothesis test between $u$ and $u+\delta$~\cite{bell1996explicit}, giving
\begin{align}
	{\bf{E}}[{\epsilon_u}^2]
	& \triangleq\frac{1}{2}\int_0^{{u_s}} {\Pr \left( {\left| {\hat u - {u_0}} \right| \ge \frac{\delta }{2}} \right)\delta\, d\delta } \notag\\
	&\ge \int_0^{{u_s}} A(\delta)\,\Phi \left( {\frac{\delta}{2}\sqrt J } \right)\delta\,d\delta,
	\label{ZZB}
\end{align}
where $\Phi(z)=\int_z^\infty \frac{1}{\sqrt{2\pi}}e^{-v^2/2}dv$ is the Gaussian $Q$-function, $A(\delta)\triangleq\int_{-\infty}^{+\infty}\min[p(u),p(u+\delta)]\,du$ is the ambiguity between the two hypotheses, $J$ denotes the Fisher information, and $u_s$ is the mainlobe width.

Denote the normalized beampattern correlation as
\begin{equation}
	G(\delta)\triangleq\frac{1}{{{N^2}}}{\left| {{{\bf{\alpha }}^H}({\bf{x}},u){\bf{\alpha }}({\bf{x}},u + \delta )} \right|^2},
\end{equation}
where $G(0)=1$ at the mainlobe peak and $-G''(0)>0$ represents the curvature of the beampattern at this peak. 
By definition, the mainlobe width $u_s$ satisfies $G(u_s)=0$. Using the Taylor expansion $G(\delta)\approx1+\frac{1}{2}G''(0)\delta^2$, we have 
\begin{equation}
	{u_s} \approx \sqrt {\frac{2}{{ - G''(0)}}}.
	\label{us}
\end{equation}

Governed by the same curvature, the Fisher information is 
\begin{align}
	J &= - T\gamma \cdot G''(0) \notag\\
	&=T\gamma\frac{2}{N}\left(\frac{2\pi}{\lambda}\right)^2{{\bf{x}}^{\rm T}}\left( {{{\bf{I}}_N} - \frac{1}{N}{\bf{1}}{{\bf{1}}^{\rm T}}} \right){\bf{x}},
	\label{J}
\end{align}
where $\gamma=\frac{N^2{\rm{SNR}}^2}{1+N{\rm{SNR}}}$ and ${\bf 1}$ is an $N\times1$ all-ones vector.

Within the mainlobe ($\delta<u_s$), $p(u)$ is approximately uniform, so $A(\delta)\approx1$ \cite{khan2009explicit}. 
Integrating \eqref{ZZB} by parts,
\begin{align}
	{\bf{E}}[{\epsilon_u}^2]
	&\approx\frac{1}{2}\left. {\Phi \left( {\frac{\delta}{2}\sqrt J } \right){\delta^2}} \right|_0^{{u_s}} - \frac{1}{2}\int_0^{{u_s}} {{\delta^2}\,d\Phi \left( {\frac{\delta}{2}\sqrt J } \right)} \notag \\
	&=\frac{1}{2}\Phi \left( {\frac{{{u_s}}}{2}\sqrt J } \right){u_s}^2 + {J^{ - 1}} \cdot \frac{2}{{\sqrt \pi  }}\int_0^{\frac{{{u_s}^2J}}{8}} {\sqrt v\, {e^{ - v}}\,dv} \notag \\
	&=\frac{1}{2}\Phi \left( {\frac{{{u_s}}}{2}\sqrt J } \right){u_s}^2 + {J^{ - 1}}{\Gamma _{3/2}}\left( {\frac{{{u_s}^2J}}{8}} \right),
	\label{integral_cal}
\end{align}
where $\Gamma_{3/2}(z)=\frac{2}{\sqrt{\pi}}\int_0^z v^{1/2}e^{-v}dv$ is the regularized lower incomplete gamma function.
Using \eqref{J} and \eqref{us} yields
\begin{equation}
	\frac{u_s}{2}\sqrt{J}=\sqrt{\frac{T\gamma}{2}},
\end{equation} 
which reduces \eqref{integral_cal} to
\begin{equation}
	{\bf{E}}[{\epsilon_u}^2] = \frac{{T\gamma }}{J}\Phi \left( {\sqrt {\frac{{T\gamma }}{2}} } \right) + {J^{ - 1}}{\Gamma _{3/2}}\left( {\frac{{T\gamma }}{4}} \right).
	\label{integral_final}
\end{equation}
Substituting \eqref{Pn_expression}, \eqref{us} and \eqref{integral_final} into \eqref{MSE_split} yields the closed-form MSE approximation in \eqref{MSE_ap}.

Next, we derive the high-SNR asymptote of \eqref{MSE_ap} to establish \eqref{MSE_theorem_highSNR_ap}. 
Applying the Marcum $Q$-function bound $Q(\alpha,\beta)\le\exp[-(\beta-\alpha)^2/2]$ and the Bessel-function asymptotic $I_m(x)\sim e^{x}/\sqrt{2\pi x}$ as $x\to\infty$ to \eqref{Pn_expression}, the two terms comprising $P_n$ both vanish exponentially in $S$: 
\begin{align}
	&Q\!\left(\!\sqrt{\tfrac{S}{2}(1\!-\!\sqrt{1\!-\!r_n^2})},\sqrt{\tfrac{S}{2}(1\!+\!\sqrt{1\!-\!r_n^2})}\!\right)
	\le e^{-\frac{S\sqrt{1-r_n^2}}{2}} \to 0,\notag\\
	&e^{-\frac{S}{2}}I_m\!\left(\frac{r_nS}{2}\right)
	\sim \frac{e^{-(1-r_n)S/2}}{\sqrt{\pi r_nS}}\to 0,\notag
\end{align}
as $S\to\infty$, hence $\lim_{S\to\infty}P_n=0$, i.e., the outlier probability is negligible at high SNR.
Meanwhile, since $\Phi(z)\to0$ and $\Gamma_{3/2}(z)\to1$ as $z\to\infty$ \cite{khan2009explicit}, \eqref{integral_final} reduces to $\mathbb{E}[\epsilon_u^2]\to J^{-1}$. 
Substituting this together with $\lim_{S\to\infty}P_n=0$ into \eqref{MSE_split} gives $\lim_{{\rm SNR}\to\infty}{\rm MSE}({\bf x})=J^{-1}$; substituting \eqref{J} into this limit yields the CRB expression in \eqref{MSE_theorem_highSNR_ap}, completing the proof.

\section{}\label{appendix_sepMSE}

Substituting the MTA geometry \eqref{MT_geo} into the general result \eqref{MSE_ap} of Theorem~\ref{Theorem_MSE}, the MSE of MTA is
\begin{align}
	{\rm{MS}}{{\rm{E}}_{{\rm{MT}}}}
	\approx\;& P_0^{{\rm{MT}}} \left[ {\frac{{T\gamma }}{{{J_{{\rm{MT}}}}}}\Phi \left( {\sqrt {\frac{{T\gamma }}{2}} } \right)+ {J_{{\rm{MT}}}}^{ - 1}{\Gamma _{3/2}}\left( {\frac{{T\gamma }}{4}} \right)} \right] \notag\\
	&+ \sum\limits_{m\neq0}  {P_m^{{\rm{MT}}}{{({u_m} - {u_0})}^2}}, \label{MSE_SP_apdx}
\end{align}
where $P_0^{\rm{MT}}$ and $P_m^{\rm{MT}}$ denote the probabilities of selecting the mainlobe and the $m$-th sidelobe respectively. $m=\pm1,\pm2,\ldots,\pm\left\lfloor\frac{D_{\rm{sep}}}{\lambda}\right\rfloor$, with $m>0$ ($m<0$) corresponding to lobes on the right (left) of the mainlobe, and $u_{m}$ its location.
$J_{\rm{MT}}$ is the Fisher information of MTA.

For even $N$, the MTA consists of two $N_1=N_2=N/2$-element half-wavelength ULAs separated by $D_{\rm{sep}}$.
Denoting their steering vectors by ${\boldsymbol{\alpha}}_{\rm{MT},1}(u)$ and ${\boldsymbol{\alpha}}_{\rm{MT},2}(u)$, \eqref{array vector} gives
\begin{equation}
	{{\boldsymbol{\alpha}}_{{\rm{MT}},2}}(u) = {e^{j\frac{{2\pi }}{\lambda }{D_{\rm{sep}}}u}}{{\boldsymbol{\alpha}}_{{\rm{MT}},1}}(u),
\end{equation}
so the beam pattern of MTA is the coherent sum of the two subarray patterns.
A sidelobe occurs when the two subarrays realign in phase, i.e., $\frac{2\pi}{\lambda}D_{\rm{sep}}(u-u_0)=2\pi m$, giving
\begin{align}
	&{u_m} = \mathop {\arg }\limits_{u \in [ - 1,1]} \left[ {{e^{j\frac{{2\pi }}{\lambda }{D_{\rm{sep}}}(u-u_0)}} = 1} \right]\\
	\Leftrightarrow&{u_{m}} = u_0+\frac{m\lambda}{D_{\rm{sep}}}.
	\label{uk}
\end{align}
Because $u_m\in[-1,1]$, the total number of sidelobe is
\begin{equation}
	{N_p^{\rm{MT}}} = 2\left\lfloor {\frac{{{D_{\rm{sep}}}}}{\lambda }} \right\rfloor.
	\label{np}
\end{equation}

Each subarray is a dense ULA, whose AoA estimate is asymptotically Gaussian-distributed about $u_0$ with variance equal to its own CRB $B_{\rm{MT}}$ \cite{van2004detection}.
As the sidelobes are spaced by $\lambda/D_{\rm{sep}}$ \eqref{uk}, the $m$-th sidelobe is selected when this Gaussian estimate falls within $[u_m-\frac{\lambda}{2D_{\rm{sep}}},u_m+\frac{\lambda}{2D_{\rm{sep}}}]$, giving
\begin{align}
	P_m^{{\rm{MT}}} = \int_{{u_m} - \frac{\lambda}{2{D_{\rm{sep}}}}}^{{u_m} + \frac{\lambda }{2{D_{\rm{sep}}}}} {\frac{1}{{\sqrt {2\pi {B_{{\rm{MT}}}}} }}} \exp \left[ { - \frac{{(u - {u_0})}^2}{{2{B_{{\rm{MT}}}}}}} \right]du\notag\\
	= \Phi \left( {\frac{{(2m - 1)\lambda }}{{2{D_{\rm{sep}}}\sqrt {{B_{{\rm{MT}}}}} }}} \right) - \Phi \left( {\frac{{(2m + 1)\lambda }}{{2{D_{\rm{sep}}}\sqrt {{B_{{\rm{MT}}}}} }}} \right),
	\label{pk_sp}
\end{align}
where $\Phi(z)=\int_z^\infty \frac{1}{\sqrt{2\pi}}e^{-v^2/2}dv$ is the Gaussian $Q$-function.
The mainlobe probability ($m=0$) is thus
\begin{equation}
	P_0^{{\rm{MT}}} = 1 - 2\Phi \left( \frac{\lambda }{{2{D_{\rm{sep}}}\sqrt {{B_{{\rm{MT}}}}} }} \right).
	\label{p0_sp}
\end{equation}

Substituting the CRB expression \eqref{CRB_x} for a single $N_1$-element subarray under effective SNR $\gamma_1\triangleq\frac{N_1^2{\rm{SNR}}^2}{1+N_1{\rm{SNR}}}$ yields
\begin{align}
	{B_{{\rm{MT}}}} 
	&= \frac{{{\lambda ^2}}}{{8{\pi ^2}T\gamma_1}}\frac{{{N_1}}}{{\frac{{{\lambda ^2}}}{4}\frac{{{N_1}({N_1}^2 - 1)}}{{12}}}}\notag\\
	&\overset{(a)}{\approx} \frac{{6}}{{{\pi ^2}T\gamma_1{N_1}^2}}\\
	&\overset{(b)}{=} \frac{{48(2 + N{\rm{SNR}})}}{{{\pi ^2}T{N^4}{\rm{SN}}{{\rm{R}}^2}}}, \label{CRB_SP} 
\end{align}
where $(a)$ holds for $N_1^2\gg1$, and $(b)$ substitutes $N_1=N/2$.

Next, we analyse the Fisher information of MTA, $J_{\rm{MT}}$.
Substituting \eqref{MT_geo} into \eqref{J} and writing $\bar{\bf{x}}\triangleq\frac1N\sum_{j=1}^N{\bf{x}}_{\rm{MT}}(j)$ for the array centroid,
\begin{equation}
	{J_{{\rm{MT}}}} = T\gamma  \cdot \frac{2}{N}{\left( {\frac{{2\pi }}{\lambda }} \right)^2}\sum\limits_{i = 1}^N {{{\left( {{{\bf{x}}_{{\rm{MT}}}}(i) - {\bf{\bar x}}} \right)}^2}}.
	\label{J_SP_def}
\end{equation}
Since the two $N_1$-element subarrays are symmetric about $\bar{\bf{x}}$, the position-variance sum can be split into the variance of a single half-wavelength ULA and an inter-subarray term, giving
\begin{align}
	\sum\limits_{i = 1}^N {{{\left( {{{\bf{x}}_{{\rm{MT}}}}(i) - {\bf{\bar x}}} \right)}^2}}
	&=\frac{\lambda^2N_1(N_1^2-1)}{24}+\frac{ND_{\rm{sep}}^2}{4}\notag\\
	&\approx\frac{\lambda^2N^3}{192}+\frac{ND_{\rm{sep}}^2}{4},
	\label{var_MT}
\end{align}
where the approximation drops the ${\cal O}(N)$ term relative to the leading ${\cal O}(N^3)$ term.
Substituting \eqref{var_MT} into \eqref{J_SP_def} gives
\begin{equation}
	J_{\rm{MT}} = T\gamma\,\pi^2\left(\frac{N^2}{24}+\frac{2D_{\rm{sep}}^2}{\lambda^2}\right). \label{J_SP} 
\end{equation}

From \eqref{uk}, $u_m-u_0=m\lambda/D_{\rm{sep}}$ is odd in $m$, while $B_{\rm{MT}}$ in \eqref{CRB_SP} does not depend on the sign of $m$, so $P_m^{{\rm{MT}}}$ in \eqref{pk_sp} is even in $m$, i.e., $P_{-m}^{{\rm{MT}}}=P_m^{{\rm{MT}}}$.

Hence the two-sided sidelobe sum in \eqref{MSE_SP_apdx} pairs up into a one-sided sum indexed by $k\triangleq|m|$, expressed as 
\begin{equation}
	\sum_{m\neq0}P_m^{{\rm{MT}}}(u_m-u_0)^2 = \frac{2\lambda^2}{D_{\rm{sep}}^2}\sum_{k=1}^{\lfloor D_{\rm{sep}}/\lambda\rfloor}k^2P_k^{{\rm{MT}}}.
	\label{fold_sum}
\end{equation}
Substituting\eqref{CRB_SP}, \eqref{J_SP} and\eqref{fold_sum} into \eqref{MSE_SP_apdx} yields the analytical MSE expression of the MTA summarized in \eqref{MSE_SP_coro}.

\end{appendices}

\bibliography{MA_edit}

@article{you2025next,
	title={Next generation advanced transceiver technologies for {6G} and beyond},
	author={You, Changsheng and Cai, Yunlong and Liu, Yuanwei and Di Renzo, Marco and Duman, Tolga M and Yener, Aylin and Swindlehurst, A Lee},
	fjournal={IEEE Journal on Selected Areas in Communications},
	journal = {IEEE J. Sel. Areas Commun.},
	volume={43},
	number={3},
	pages={582--627},
	year={2025},
	month=Mar,
	publisher={IEEE}
}

@article{liu2022survey,
	title={A survey on fundamental limits of integrated sensing and communication},
	author={Liu, An and Huang, Zhe and Li, Min and Wan, Yubo and Li, Wenrui and Han, Tony Xiao and Liu, Chenchen and Du, Rui and Tan, Danny Kai Pin and Lu, Jianmin and Shen, Yuan and Colone, Fabiola and Chetty, Kevin},
	journal={IEEE Commun. Surv. Tutor.},
	fjournal={IEEE Communications Surveys \& Tutorials},
	volume={24},
	number={2},
	pages={994--1034},
	year={2022},
	month = {Secondquarter},
	publisher={IEEE}
}

@article{song2023intelligent,
	title={Intelligent reflecting surface enabled sensing: {Cram\'er-Rao} bound optimization},
	author={Song, Xianxin and Xu, Jie and Liu, Fan and Han, Tony Xiao and Eldar, Yonina C},
	journal={IEEE Trans. Signal Process.},
	fjournal={IEEE Transactions on Signal Processing},
	volume={71},
	pages={2011--2026},
	year={2023},
	month=Mar,
	publisher={IEEE}
}

@article{lu2014overview,
	title={An overview of massive {MIMO}: Benefits and challenges},
	author={Lu, Lu and Li, Geoffrey Ye and Swindlehurst, A Lee and Ashikhmin, Alexei and Zhang, Rui},
	journal={IEEE J. Sel. Top. Signal Process.},
	fjournal={IEEE Journal of Selected Topics in Signal Processing},
	volume={8},
	number={5},
	pages={742--758},
	year={2014},
	month= Oct,
	publisher={IEEE}
}

@article{wang2024tutorial,
	title={A tutorial on extremely large-scale {MIMO} for {6G}: Fundamentals, signal processing, and applications},
	author={Wang, Zhe and Zhang, Jiayi and Du, Hongyang and Niyato, Dusit and Cui, Shuguang and Ai, Bo and Debbah, M{\'e}rouane and Letaief, Khaled B and Poor, H Vincent},
	journal={IEEE Commun. Surv. Tutor.},
	fjournal={IEEE Communications Surveys \& Tutorials},
	volume={26},
	number={3},
	pages={1560--1605},
	year={2024},
	month={thirdquarter} ,
	publisher={IEEE}
}

@article{roberts2011sparse,
	title={Sparse antenna array design for {MIMO} active sensing applications},
	author={Roberts, William and Xu, Luzhou and Li, Jian and Stoica, Petre},
	fjournal={IEEE Transactions on Antennas and Propagation},
	journal={IEEE Trans. Antennas Propag.},
	volume={59},
	number={3},
	pages={846--858},
	year={2011},
	month= Mar,
	publisher={IEEE}
}

@article{cen2009linear,
	title={Linear sparse array synthesis with minimum number of sensors},
	author={Cen, Ling and Ser, Wee and Yu, Zhu Liang and Rahardja, Susanto and Cen, Wei},
	fjournal={IEEE Transactions on Antennas and Propagation},
	journal={IEEE Trans. Antennas Propag.},
	volume={58},
	number={3},
	pages={720--726},
	year={2010},
	month= Mar,
	publisher={IEEE}
}

@inproceedings{wang2023can,
	title={Can sparse arrays outperform collocated arrays for future wireless communications?},
	author={Wang, Huizhi and Zeng, Yong},
	booktitle={Proc. 2023 IEEE Globecom Workshops (GC Wkshps)},
	pages={667--672},
	year={2023},
	organization={IEEE},
	address={Kuala Lumpur, Malaysia}
}

@article{zhu2024movable,
	title={Movable Antenna Enabled Near-Field Communications: Channel Modeling and Performance Optimization},
	author={Lipeng Zhu and Wenyan Ma and Zhenyu Xiao and Rui Zhang},
	fjournal={IEEE Transactions on Communications},
	journal={IEEE Trans. Commun.},
	year={2025},
	month= Sep,
	volume={73},
	number={9},
	pages={7240-7256}
}

@ARTICLE{ma2026movabletsp,
	author={Ma, Wenyan and Zhu, Lipeng and Zhang, Rui},
	fjournal={IEEE Transactions on Signal Processing}, 
	journal   = {IEEE Trans. Signal Process.},
	title={Movable Antenna Enhanced Integrated Sensing and Communication Via Antenna Position Optimization}, 
	year={2026},
	month= Sep,
	volume={74},
	number={},
	pages={1522-1537},
}

@article{mao2025movable,
	author={H. Mao and L. Zhu and W. Ma and Z. Xiao and X.-G. Xia and R. Zhang},
	title={Movable-Antenna Array Enhanced Multi-Target Sensing: {CRB} Characterization and Optimization},
	journal={arXiv preprint arXiv:2511.18907},
	year={2025}
}

@article{wang2025movable,
	author={Y. Wang and W. Mei and X. Wei and Y.-F. Wu and Z. Chen and B. Ning},
	title={Movable Antenna Empowered Near-Field Sensing via Antenna Position Optimization},
	journal={arXiv preprint arXiv:2512.00758},
	year={2025}
}

@book{cinlar2013introduction,
	author    = {Cinlar, Erhan},
	title     = {Introduction to Stochastic Processes},
	edition   = {{Dover} reprint},
	series    = {Dover Books on Mathematics Series},
	address   = {Mineola, NY, USA},
	publisher = {Dover Publications},
	year      = {2013},
	isbn      = {978-0-486-49797-6}
}

@book{van2004detection,
	author    = {Van Trees, Harry L.},
	title     = {Detection, Estimation, and Modulation Theory, Part {I}: Detection, Estimation, and Linear Modulation Theory},
	address   = {Hoboken, NJ, USA},
	publisher = {John Wiley \& Sons},
	year      = {2004},
}

@Article{athley2005threshold,
	author    = {Athley, Fredrik},
	journal   = {IEEE Trans. Signal Process.},
	title     = {Threshold region performance of maximum likelihood direction of arrival estimators},
	year      = {2005},
	month= Apr,
	number    = {4},
	pages     = {1359--1373},
	volume    = {53},
	fjournal  = {IEEE Transactions on Signal Processing},
	publisher = {IEEE},
}

@book{proakis2001digital,
	author    = {Proakis, John G. and Salehi, Masoud},
	title     = {Digital Communications},
	edition   = {4th},
	address   = {Boston, MA, USA},
	publisher = {McGraw-Hill},
	year      = {2001},
	isbn      = {978-0-07-232111-1}
}

@inproceedings{khan2009explicit,
	title={Explicit {Ziv-Zakai} bound for {DOA} estimation with sparse linear arrays},
	author={Khan, Diba and Bell, Kristine L},
	booktitle={Proc. 2009 3rd IEEE International Workshop on Computational Advances in Multi-Sensor Adaptive Processing (CAMSAP)},
	pages={257--260},
	year={2009},
	address={Aruba, Netherland Antilles},
	organization={IEEE}
}

@article{khan2013explicit,
	title={Explicit {Ziv-Zakai} bound for analysis of {DOA} estimation performance of sparse linear arrays},
	author={Khan, Diba and Bell, Kristine L},
	fjournal={Signal processing},
	journal={Signal Process.},
	volume={93},
	number={12},
	pages={3449--3458},
	year={2013},
	publisher={Elsevier}
}

@article{ziv1969some,
	title={Some lower bounds on signal parameter estimation},
	author={Ziv, Jacob and Zakai, Moshe},
	journal={IEEE Trans. Inf. Theory},
	volume={15},
	number={3},
	pages={386--391},
	year={1969},
	month=May,
	publisher={IEEE}
}

@article{weiss1985lower,
	title={A lower bound on the mean-square error in random parameter estimation (corresp.)},
	author={Weiss, A and Weinstein, Ehud},
	fjournal={IEEE Transactions on information theory},
	journal   = {IEEE Trans. Inf. Theory},
	volume={31},
	number={5},
	pages={680--682},
	year={1985},
	month=Sep,
	publisher={IEEE}
}

@Article{bell1997extended,
	author    = {Bell, Kristine L and Steinberg, Yossef and Ephraim, Yariv and Van Trees, Harry L},
	journal   = {IEEE Trans. Inf. Theory},
	title     = {Extended {Ziv-Zakai} lower bound for vector parameter estimation},
	year      = {1997},
	month=Mar,
	number    = {2},
	pages     = {624--637},
	volume    = {43},
	fjournal  = {IEEE Transactions on Information Theory},
	publisher = {IEEE},
}

@Article{bell1996explicit,
	author   = {Bell, K.L. and Ephraim, Y. and Van Trees, H.L.},
	journal  = {IEEE Trans. Signal Process.},
	title    = {Explicit {Ziv-Zakai} lower bound for bearing estimation},
	year     = {1996},
	month=Nov,
	number   = {11},
	pages    = {2810-2824},
	volume   = {44},
	fjournal = {IEEE Transactions on Signal Processing},
}

@inproceedings{khan2010analysis,
	title={Analysis of {DOA} estimation performance of sparse linear arrays using the {Ziv-Zakai} bound},
	author={Khan, Diba and Bell, Kristine L},
	booktitle={Proc. 2010 IEEE Radar Conference},
	pages={746--751},
	year={2010},
	address={Arlington, VA, USA,},
	organization={IEEE}
}

@phdthesis{athley2003space,
	author      = {F. Athley},
	title       = {Space-time parameter estimation in radar array processing},
	school      = {Dept. Signals and Systems, Chalmers Univ. Technology},
	address     = {Gothenburg, Sweden},
	year        = {2003}
}

@Article{athley2007radar,
	author    = {Athley, Fredrik and Engdahl, Christer and Sunnergren, Per},
	journal   = {IEEE Trans. Aerosp. Electron. Syst.},
	title     = {On radar detection and direction finding using sparse arrays},
	year      = {2007},
	month=Oct,
	number    = {4},
	pages     = {1319--1333},
	volume    = {43},
	fjournal  = {IEEE Transactions on Aerospace and Electronic Systems},
	publisher = {IEEE},
}

@article{ma2024movable,
	title={Movable antenna enhanced wireless sensing via antenna position optimization},
	author={Ma, Wenyan and Zhu, Lipeng and Zhang, Rui},
	journal= {IEEE Trans. Wireless Commun.},
	volume={23},
	number={11},
	pages={16575--16589},
	year={2024},
	month=Nov,
	fjournal  = {IEEE Transactions on Wireless Communications},
	publisher={IEEE}
}

@article{zhu2026tutorial,
	title={A tutorial on movable antennas for wireless networks},
	author={Zhu, Lipeng and Ma, Wenyan and Mei, Weidong and Zeng, Yong and Wu, Qingqing and Ning, Boyu and Xiao, Zhenyu and Shao, Xiaodan and Zhang, Jun and Zhang, Rui},
	journal={IEEE Commun. Surv. Tutor.},
	fjournal={IEEE Communications Surveys \& Tutorials},
	volume={28},
	pages={3002--3054},
	year={2026},
	publisher={IEEE}
}

@article{shao20246d,
	title={{6D} movable antenna based on user distribution: Modeling and optimization},
	author={Shao, Xiaodan and Jiang, Qijun and Zhang, Rui},
	fjournal={IEEE Transactions on Wireless Communications},
	journal = {IEEE Trans. Wireless Commun.},
	volume={24},
	number={1},
	pages={355--370},
	year={2025},
	month= Jan, 
	publisher={IEEE}
}

@article{zhu2024movabletcm,
	title={Movable antennas for wireless communication: Opportunities and challenges},
	author={Zhu, Lipeng and Ma, Wenyan and Zhang, Rui},
	fjournal={IEEE Communications Magazine},
	journal={IEEE Commun. Mag.},
	volume={62},
	number={6},
	pages={114--120},
	year={2024},
	month=Jun,
	publisher={IEEE}
}

@article{ma2023mimo,
	title={{MIMO} capacity characterization for movable antenna systems},
	author={Ma, Wenyan and Zhu, Lipeng and Zhang, Rui},
	journal={IEEE Trans. Wireless Commun.},
	fjournal={IEEE Transactions on Wireless Communications},
	volume={23},
	number={4},
	pages={3392--3407},
	year={2024},
	month=Apr,
	publisher={IEEE}
}

@article{zhu2024movabletwc,
	title={Movable-antenna enhanced multiuser communication via antenna position optimization},
	author={Zhu, Lipeng and Ma, Wenyan and Ning, Boyu and Zhang, Rui},
	journal={IEEE Trans. Wireless Commun.},
	fjournal={IEEE Transactions on Wireless Communications},
	volume={23},
	number={7},
	pages={7214--7229},
	year={2024},
	month=Jul,
	publisher={IEEE}
}

@Article{mei2024movable,
	author    = {Mei, Weidong and Wei, Xin and Ning, Boyu and Chen, Zhi and Zhang, Rui},
	journal   = {IEEE Wireless Commun. Lett.},
	title     = {Movable-antenna position optimization: A graph-based approach},
	year      = {2024},
	month=Jul,
	number    = {7},
	pages     = {1853--1857},
	volume    = {13},
	fjournal  = {IEEE Wireless Communications Letters},
	publisher = {IEEE},
}

@article{shao20256d,
	title={{6D} movable antenna enhanced wireless network via discrete position and rotation optimization},
	author={Shao, Xiaodan and Zhang, Rui and Jiang, Qijun and Schober, Robert},
	fjournal={IEEE Journal on Selected Areas in Communications},
	journal={IEEE J. Sel. Areas Commun.},
	volume={43},
	number={3},
	pages={674--687},
	year={2025},
	month=Mar,
	publisher={IEEE}
}

@article{xiao2024channel,
	title={Channel estimation for movable antenna communication systems: a framework based on compressed sensing},
	author={Xiao, Zhenyu and Cao, Songqi and Zhu, Lipeng and Liu, Yanming and Ning, Boyu and Xia, Xiang-Gen and Zhang, Rui},
	journal={IEEE Trans. Wireless Commun.},
	fjournal={IEEE Transactions on Wireless Communications},
	volume={23},
	number={9},
	pages={11814--11830},
	year={2024},
	month=Sep,
	publisher={IEEE}
}

@book{van2002modulation,
	author    = {Van Trees, Harry L.},
	title     = {Detection, Estimation, and Modulation Theory, Part {IV}: Optimum Array Processing},
	publisher = {John Wiley \& Sons},
	address   = {New York, NY, USA},
	year      = {2002},
	isbn      = {978-0-471-09390-9}
}

@book{kay1993fundamentals,
	author    = {Kay, Steven M.},
	title     = {Fundamentals of Statistical Signal Processing: Estimation Theory},
	publisher = {Prentice Hall},
	address   = {Upper Saddle River, NJ, USA},
	year      = {1993}
}

@Article{moffet1968minimum,
	author    = {Moffet, Alan},
	journal   = {IEEE Trans. Antennas Propag.},
	title     = {Minimum-redundancy linear arrays},
	year      = {1968},
	month=Mar,
	number    = {2},
	pages     = {172--175},
	volume    = {16},
	fjournal  = {IEEE Transactions on Antennas and Propagation},
	publisher = {IEEE},
}

@Article{pal2010nested,
	author    = {Pal, Piya and Vaidyanathan, Palghat P},
	journal   = {IEEE Trans. Signal Process.},
	title     = {Nested arrays: {A} novel approach to array processing with enhanced degrees of freedom},
	year      = {2010},
	month= Aug,
	number    = {8},
	pages     = {4167--4181},
	volume    = {58},
	fjournal  = {IEEE Transactions on Signal Processing},
	publisher = {IEEE},
}

@Article{vertatschitsch1991impact,
	author   = {Vertatschitsch, E.J. and Haykin, S.},
	journal  = {IEEE Trans. Antennas Propag.},
	title    = {Impact of linear array geometry on direction-of-arrival estimation for a single source},
	year     = {1991},
	month=May,
	number   = {5},
	pages    = {576-584},
	volume   = {39},
	fjournal = {IEEE Transactions on Antennas and Propagation},
}

@inproceedings{athley2000direction,
	title={Direction-of-arrival estimation using separated subarrays},
	author={Athley, Fredrik and Engdahl, Christer},
	booktitle={Proc. 34th Asilomar Conference on Signals, Systems and Computers (Cat. No. 00CH37154)},
	volume={1},
	pages={585--589},
	year={2000},
	address={Pacific Grove, CA, USA},
	organization={IEEE}
}

@inproceedings{athley2001optimization,
	title={Optimization of element positions for direction finding with sparse arrays},
	author={Athley, Fredrik},
	booktitle={Proc. 11th IEEE Signal Processing Workshop on Statistical Signal Processing (Cat. No. 01TH8563)},
	pages={516--519},
	year={2001},
	address={Singapore},
	organization={IEEE}
}

@article{ma2026movable,
	title={Movable Antennas for Robust Wireless Sensing via Joint {Cram\'er-Rao} Bound and Sidelobe Minimization},
	author={Ma, Wenyan and Zhu, Lipeng and Zhai, Weitong and Zhang, Rui},
	journal={arXiv preprint arXiv:2606.23154},
	year={2026}
}

@article{shao20256dma,
	title={{6DMA} enhanced wireless network with flexible antenna position and rotation: Opportunities and challenges},
	author={Shao, Xiaodan and Zhang, Rui},
	fjournal={IEEE Communications Magazine},
	journal={IEEE Commun. Mag.},
	volume={63},
	number={4},
	pages={121--128},
	year={2025},
	month=Apr,
	publisher={IEEE}
}
\bibliographystyle{IEEEtran}
\end{document}